\documentclass[a4paper]{article}
\usepackage[T1]{fontenc}
\usepackage[utf8]{inputenc}
\usepackage{amsmath,amssymb,amsfonts,mathtools}
\usepackage[amsmath,amsthm,thmmarks,hyperref]{ntheorem}
\newcommand{\yestag}{\stepcounter{equation}\tag{\theequation}}
\newcommand{\yesnumber}{\yestag}

\usepackage{enumitem}
\usepackage[normalem]{ulem}
\usepackage[usenames,dvipsnames]{pstricks}
\usepackage{color}

\usepackage{physics}
\usepackage{hyperref}
\hypersetup
    {pdfpagemode=UseNone
    ,pdfstartview=FitH
    ,pdfview=FitH
    ,pdfremotestartview=FitH
    ,final}
\usepackage{cleveref}

\usepackage{tikz}
\usetikzlibrary{arrows,automata}
\usetikzlibrary{positioning}
\usetikzlibrary{decorations.pathmorphing}
\usetikzlibrary{shapes.geometric}
\usetikzlibrary{fit}
\usetikzlibrary{backgrounds}

\newtheorem{theorem}             {Theorem}[section]
\newtheorem{lemma}      [theorem]{Lemma}
\newtheorem{corollary}  [theorem]{Corollary}
\newtheorem{definition} [theorem]{Definition}
\newtheorem{remark}     [theorem]{Remark}

\title{Semantic Security for Quantum Wiretap Channels }
\author{Holger Boche\\
Lehrstuhl f\"ur Theoretische Informationstechnik,\\
Technische Universit\"at M\"unchen,\\
Munich Center for Quantum Science and Technology (MCQST) \\
Munich, Germany\\
CASA: Cyber Security in the Age\\ of Large-Scale Adversaries Exzellenzcluster,\\ Ruhr-Universität Bochum,\\
Bochum, Germany\\
boche@tum.de
\and
Minglai Cai\\
Grup d’Informaci\'o Qu\`antica, \\
Universitat Aut\`onoma de Barcelona,\\
Barcelona, Spain\\
minglai.cai@tum.de 
\and 
Christian Deppe \\
Lehr- und Forschungseinheit f\"ur Nachrichtentechnik\\ 
Technische Universit\"at M\"unchen,\\
Munich, Germany\\
christian.deppe@tum.de  
\and 
Roberto Ferrara \\
Lehr- und Forschungseinheit f\"ur Nachrichtentechnik\\ 
Technische Universit\"at M\"unchen,\\
Munich, Germany\\
roberto.ferrara@tum.de 
\and
Moritz Wiese \\
Lehrstuhl f\"ur Theoretische
Informationstechnik,\\
Technische Universit\"at M\"unchen,\\
Munich, Germany\\
CASA: Cyber Security in the Age\\ of Large-Scale Adversaries Exzellenzcluster,\\ Ruhr-Universität Bochum,\\
Bochum, Germany\\
wiese@tum.de
}

\begin{document}

\maketitle
\begin{abstract}
We consider the problem of semantic security via classical-quantum and quantum wiretap channels and use explicit constructions to transform a non-secure code into a semantically secure code achieving capacity by means of biregular irreducible functions (BRI functions). 
Explicit parameters in the finite regimes can be extracted from the theorems.
We also generalize the semantic security capacity theorem, which shows that a strongly secure code guarantees a semantically secure code with the same secrecy rate, to any quantum channel, including infinite-dimensional and non-Gaussian ones.
%
\end{abstract}


\section{Introduction}
We investigate the transmission of messages from a sending to a 
receiving party through a wiretap channel. In this model, there is a third party called an eavesdropper who must not be allowed to know the information sent from the sender to the intended receiver. The wiretap channel was first introduced by Wyner in~\cite{Wyn}. A classical-quantum channel with an eavesdropper is called a classical-quantum wiretap channel.

The secrecy capacity of the classical-quantum wiretap channel subject to the strong security criterion has been determined in~\cite{De,Ca/Wi/Ye}. Strong security means that given a  uniformly distributed message sent through the channel, the eavesdropper shall obtain no information about it. This criterion goes back to~\cite{Csi, Mau} and it is the most common secrecy criterion in classical and quantum information theory.

In the present paper, however, a stronger security requirement will be applied, called semantic security (and defined in \Cref{sect:strong_sem_sec}). 
With this, the eavesdropper gains no information regardless of the message distribution. 
This criterion was introduced to information theory from cryptography~\cite{Be/Te/Va} and also used earlier in QKD in~\cite[Theorem~2]{Hayashi07},
motivated by the analogous security criterion of the same name.
It is equivalent to  message indistinguishability, where the eavesdropper cannot distinguish whether the given cipher text is an encryption of any two messages (which can even be chosen by the eavesdropper). 
Aside from being the minimum security requirement in practical applications, 
semantic security is also necessary in the security of identification codes~\cite{ahlswede1989identification}. 
Because in identification pairs of messages are compared to each other, to make the code secure, any two messages must be indistinguishable at the eavesdropper~\cite{AZ95}. 
Message indistinguishability is thus necessary to construct secure identification codes, and the semantic security achieved in this paper can thus be used to construct secure identification codes via classical-quantum channels~\cite{Bo/De/Wi}.
At the same time but on a different note, bounds on identification find application on transmission via the wiretap channel~\cite{Hayashi06}.

In \Cref{sscocqwc}, we prove the well-known semantic-secrecy capacity formula, which is equal to the capacity formula under the strong security criterion, for the most general case of arbitrary, including infinite dimentional channel, under the most stringent semantic security criterion (in terms of mutual information).
This can be easy shown by slightly modifying the expurgation method~\cite{Ha15}, whereby a code which has a small leakage with respect to the strong security criterion is converted into a code which has a small leakage with respect to the semantic security leakage, making the strong and semantic secrecy capacities equal.
This statement was at first proven for classical wiretap channels in~\cite{Wi/Bo} and for finite dimensional classical-quantum channels in~\cite{Hay/Mat} and with security measured in trace norm in~\cite{Re/Re}.
Notice that that the expurgation technique given in~\cite{Ha15} works for any channels including infinite dimensional channel, even non-Gaussian ones.
The results for the classical-quantum channel extend to quantum channels where the environment, which  is completely under the control of a constant eavesdropper, can be entangled with the quantum system.

\begin{figure}
	\centering
	
	\includegraphics{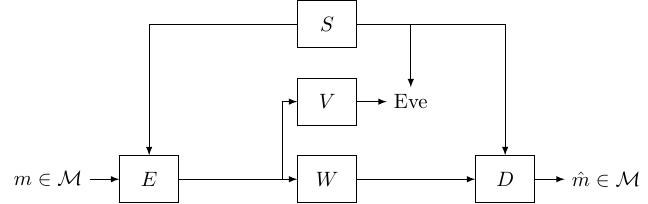}
    
    
	\caption{\label{fig:cr-scheme}A general common-randomness coding scheme. $W$ denotes the classical-quantum channel between sender and receiver, while $V$ is the wiretap. $\{\{E_m,D_m: m\in\mathcal{M}\}:s\in\mathcal{S}\}$ is a common-randomness code for $W$.}

    \bigskip

    \includegraphics{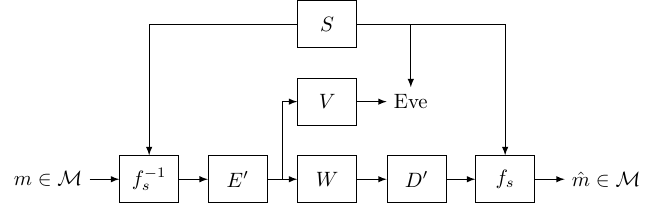}
    
    
	\caption{\label{fig:BT_scheme}The modular BRI scheme. $\mathcal{C} =\{E_m',D_m': m\in\mathcal{M}\}$ is a code for $W$ and $f$ is a biregular irreducible function. In practice, $\mathcal{C}$ will be a transmisison code, namely a code with low error. $f_s^{-1}$ denotes the random choice of an element of $f_s^{-1}(m)$ for the given message $m$ and seed $s$. The seed $s$ has to be known to the sender and receiver beforehand or it is generated by the sender and transmitted to the intended receiver through the channel.}
\end{figure}

The proofs applying the expurgation technique 
 are merely existence statements and give no clue as to how to find the large message subset which provides semantic security.
In \Cref{sscocqwcBRI}, we show how the capacity can be achieved by modularly correcting transmission error and amplifying privacy in separate components of the code, as for the case of strong secrecy~\cite{Hay_exp}.
While we make the proof for finite dimensional channels, the codes automatically also achieve capacity for quantum Gaussian wiretap channels, since the capacity is achieved as the limit of the capacities on finite-dimensional subspaces of increasing dimension~\cite{GSE08}.
These modular codes for the classical-quantum wiretap channel are constructed concatenating an ordinary transmission code for the channel from the sender to the intended receiver with an additional security component. 
Furthermore, the additional security component is independent of the channel (sometimes called channel universality~\cite{Hay/Mat}), as for the case of explicit strong-secrecy constructions.
The first such security components used in the literature were universal hash functions~\cite{Be/Br/Ro, Hay_exp}, used to achieve strong secrecy. The specific security components we use, called biregular irreducible functions (BRI functions), were introduced in~\cite{Wi/Bo} in the context of classical wiretap channels. A modular code for the classical-quantum wiretap channel is illustrated in \Cref{fig:BT_scheme}. 

If a transmission code from the sender to the intended receiver with input/output set $\mathcal C$ is given, then a BRI function that is to be used with this transmission code has the form $f:\mathcal S\times\mathcal C\to\mathcal N$. Here, $\mathcal S$ is a seed set, and the set $\mathcal M$ of messages of the modular wiretap code is an explicitly given subset of $\mathcal N$. 
To use this modular code, the sender and the intended receiver have to share a seed $s\in {\mathcal S}$, chosen uniformly at random from $\mathcal S$. Given any message $m \in {\mathcal M}$ and seed $s$, the sender randomly chooses a preimage $c\in\mathcal C$, satisfying $f_s(c)=m$. Since the intended receiver knows $s$, he can recover $m$ if no transmission error occurs. Thus the task of establishing reliable transmission is entirely due to the transmission code, while the BRI function's responsibility is to ensure semantic security. 
The above modular construction was already shown to achieve the secrecy capacity of classical wiretap channels with semantic security in~\cite{Wi/Bo}. An alternative to BRI functions was proposed by Hayashi and Matsumoto~\cite{Hay/Mat}. Their example, however, requires a seed which is longer than that which is necessary for the best-known BRI function.
The length of the seed is relevant for the efficiency of the derandomized codes at finite regimes.

We emphasize that the seed is not a secret key, since we do not require it to be unknown to the eavesdropper. The main part of the analysis of the above modular codes assumes that the seed is given (non-securely) to the sender and the intended receiver by common randomness. However, it is a general result for codes with common randomness that if the error probability and security leakage decrease sufficiently fast in block length, the seed can be reused a small number of times. Modular codes constructed using BRI functions show this behavior. Therefore, no more than a negligible amount of rate is lost if the sender generates the seed and transmits it to the intended receiver, and then reuses the seed a small number of times. In particular, any rate which is achievable with a seed given by common randomness is also achievable with a sender-generated seed.

Moreover we would like to emphasize that the semantic secrecy for classical-quantum channels is much harder than for classical channels. 
Roughly speaking, two different inputs not only result in two different random variables, but the outputs have also different eigenspace, so it is more difficult to make them indistinguishable.
For instance, \Cref{thm:bri-function} delivers a bound which is technically weaker than the classical version of \cite{Wi/Bo} (see below).


\section{Basic Notations and Definitions}
\label{BNaD}

The content of this section can be found in most books covering 
the basics of matrix analysis and quantum information, see, for example,~\cite{Bh,Ni/Ch,Wil}.
For a finite set $\mathcal{X}$, we denote the set of probability distributions on $\mathcal{X}$ by $P(\mathcal{X})$ and with $\mathbb{E}$ the expectation value.
For a finite-dimensional complex Hilbert space  $H$, 
we denote the set  of linear  operators on $H$ with $\mathcal{L}(H)$.
Let $\rho,\sigma \in \mathcal{L}(H)$ be  Hermitian operators in $\mathcal{L}(H)$.
We say $\rho\geq\sigma$, or equivalently $\sigma\leq\rho$, if $\rho-\sigma$ is positive-semidefinite.
The (convex) space of  density operators on $H$ is defined as
\[\mathcal{S}(H)\coloneqq \{\rho \in \mathcal{L}(H) :\rho \geq 0_{H},\ \tr(\rho) = 1 \},\]
where $0_{H}$ is the null matrix on $H$. 
Note that any operator in $\mathcal{S}(H)$ is bounded.
A POVM (positive-operator valued measure) over a finite set $\mathcal M$ is
a collection of positive-semidefinite operators 
$\left\{D_m: m\in \mathcal M\right\}$ on ${H}$, 
which is a partition of the identity, 
i.e. $\sum_{m\in\mathcal M} D_m=\mathrm{id}_{H}$.
The POVM describes a measurement that maps quantum states $\rho$ 
to classical values $m\in\mathcal M$ by assigning them the probability $\tr [\rho D_m]$.
If $\sum_{m\in\mathcal M} D_m \leq \mathrm{id}_{H}$ then we call $\left\{D_m: m\in \mathcal M\right\}$ a sub-POVM.
More generally, a measurement operator will be any positive semi-definite operator $D$ satisfying $0 \leq D \leq \mathrm{id}_{H}$.

\vspace{0.15cm}

For a quantum state $\rho\in \mathcal{S}(H)$, we denote its von Neumann
entropy by \[S(\rho)\coloneqq- \tr(\rho\log\rho)\text{,}\] 
and for a discrete random variable $X$,  on a finite set $\mathcal{X}$ 
we denote the Shannon entropy of $X$ by 
\begin{equation*}
    H(X)\coloneqq-\sum_{x \in \mathcal{X}}p(x)\log p(x),
\end{equation*}
where we use in both definitions and throughout this paper the convention
that the logarithm ``$\log$'' is taken in base $2$.
We denote with 
$h(\nu) \coloneqq -\nu \log \nu - (1- \nu) \log (1-\nu)$
for $\nu\in [0,1]$, the binary entropy.

Let $\rho$ and $\sigma$ be two positive semi-definite operators not necessarily in $\mathcal S(H)$. The quantum relative entropy between  $\rho$ and $\sigma$ is defined as
\[
    D(\rho\parallel\sigma) 
    \coloneqq \tr\rho\left(\log \rho- \log \sigma \right)
\]
if $supp(\rho) \subset supp(\sigma)$, and $D(\rho\parallel\sigma)\coloneqq\infty$ otherwise.
For $\alpha$ $\in(0,1)\cup(1,\infty)$, the R\'enyi relative entropy~\cite{Petz}
between  $\rho$ and $\sigma$ is defined as
\[
    D_{\alpha}(\rho\parallel\sigma) 
    \coloneqq \frac{1}{\alpha-1}\log \tr\left(\rho^{\alpha}\sigma^{1-\alpha}\right)
\]
if $supp(\rho) \subset supp(\sigma)$, and $D_\alpha(\rho\Vert\sigma)\coloneqq\infty$ otherwise.
The R\'enyi relative entropy satisfies the ordering relation, or parameter monotonicity,
\[
    D_{\alpha}(\rho\parallel\sigma)
    \leq D_{\alpha'}(\rho\parallel\sigma)
\]
for any density operators $\rho$ and $\sigma$, and any $\alpha\leq\alpha'$~\cite{Bu/Da,Mo/Hi}. Furthermore, it holds~\cite{Petz} that
\[
    \lim_{\alpha\nearrow 1} D_{\alpha}(\rho\parallel\sigma)
    = \lim_{\alpha\searrow 1} D_{\alpha}(\rho\parallel\sigma)
    = D(\rho\parallel\sigma)\text{ .}
\]
The R\'enyi relative entropies also satisfy channel monotonicity (monotonicity under quantum channels) for $\alpha\leq 2$ \cite{Ando,Mo/Hi,Au/Da}, namely under completely positive trace preserving linear maps $\Lambda$
\begin{equation*}
    D_{\alpha}(\Lambda(\rho)\parallel\Lambda(\sigma)) \leq D_{\alpha}(\rho\parallel\sigma).
\end{equation*}

For  finite-dimensional complex Hilbert spaces  $H$ and  $H'$, a \textbf{quantum channel} $N(\rho) $ is represented by a completely-positive trace-preserving linear map $N:\mathcal{L}(H) \to \mathcal{L}(H')$, which accepts input quantum states in $\mathcal{S}(H)$ and produces output quantum states in  $\mathcal{S}(H')$. Quantum channels will be treated in \Cref{sec:quantum}, building on the results for classical-quantum channels.
The case of classical-quantum channels will be treated in  \Cref{sect:strong_sem_sec,sscocqwc,sscocqwcBRI}. For a finite-dimensional complex Hilbert space $H$, a
\textbf{classical-quantum channel} is a map 
$V: \mathcal{X}\to\mathcal{S}(H)$, $ x \mapsto V(x)$.
In order to use the same notation common in classical information theory,
for a measurement operator $0\leq D\leq \mathrm{id}_H$, we define the following notation 
\begin{align*}
    \rho(D)&\coloneqq \tr(\rho\cdot D) 
    &
    V(D|x) &\coloneqq \tr(D \cdot V(x));
\end{align*}
notice that this notation, at least for the quantum state, is also common in the $C*$ algebra literature, where quantum states are considered functionals on hermitian operators.

For a probability distribution  $P$ and a classical-quantum channel $V$ on $\mathcal{X}$, the Holevo $\chi$ quantity, or Holevo information, is defined as
\begin{equation*}
    \chi(P;V)
    \coloneqq S\left(\sum_{x\in \mathcal{X}} P(x)V(x)\right)-
    \sum_{x\in \mathcal{X}} P(x)S\left(V(x)\right)\text{ .}
\end{equation*}
The Holevo information is also the mutual information between the input and the output.
Namely, let $\mathcal{H}_{\mathcal{X}}$ be 
a $|\mathcal{X}|$-dimensional Hilbert space 
with a set of orthonormal basis $\{ \ket{x} :x\in \mathcal{X}\}$,
and $X$ be the random variable on $\mathcal{X}$ with distribution $P$,
then the Holevo information is the quantum  mutual information $I(X\wedge V(X))$ for the state
\begin{equation*}
\rho_{X,V} \coloneqq 
    \sum_{x\in \mathcal{X}}  P(x) \ketbra{x} \otimes V(x).
\end{equation*}
The Holevo  quantity
can also be written as the expected value of  the 
quantum relative entropy as a  random variable 
of the states from the product states of their marginals, 
namely if we denote the marginal with 
\begin{align*}
    \rho_{X} &\coloneqq \sum_{x\in \mathcal{X}}  P(x) \ketbra{x}
&
    PV &\coloneqq \sum_{x\in \mathcal{X}}  P(x) V(x)
\\
    &= \mathbb{E}_X \ketbra{X}
&
    &= \mathbb{E}_X V(X)
\end{align*}
it is easy to check that the following holds
\begin{align}
\chi(P;V)
    \nonumber
    &= S(PV) - \sum_{x\in \mathcal{X}} P(x)S\left(V(x)\right)\\
    \label{eq:holevo-relent}
    &=D\qty(\rho_{X,V} \parallel \rho_{X} \otimes PV) \\
    &=\sum_{x\in\mathcal{X}}
      P(x) D\qty(V(x) \parallel PV)\notag\\
    \label{eq:holevo-avgrelent}
    &= \mathbb{E}_X D(V(X) \parallel \mathbb{E}_X V(X)) .
\end{align}
We then denote the conditional information of the quantum part 
conditioned on the classical one as
\begin{equation*}
    S(V|P) \coloneqq S(PV) - \chi(P;V) = 
    \sum_{x\in \mathcal{X}} P(x)S\left(V(x)\right).
\end{equation*}

\section{Strong and Semantic security}
\label{sect:strong_sem_sec}
\label{sec:semantic}

In this section we introduce the channels, the codes, and the capacities which we will study, as well as the definitions for strong secrecy and semantic security with and without common randomness. In \Cref{sec:bri} we will show how to explicitly build such codes using  modular coding schemes which require common randomness. For now, let us define the channels of interest.

\begin{definition}
Let $\mathcal{X}$ be a finite set and
$H$ and $H'$ be finite-dimensional complex Hilbert spaces. 
Let $W:\mathcal{X} \to \mathcal{S}(H)$ and 
${V}:\mathcal{X} \to \mathcal{S}(H')$ be classical-quantum channels. 
We call  the pair  $(W,{V})$ 
a \textbf{classical-quantum wiretap channel}.
\end{definition}

The intended receiver accesses the output of the first channel $W$, 
and the eavesdropper observes the output of the second channel ${V}$ in the pair.

\subsection{Codes}

A code is created by the sender and the intended receiver beforehand. The sender uses the encoder to map the
message that he wants to send to a channel input, while the intended receiver uses the family of decoder operators on the channel output to perform a measurement and decode the message. In all the definitions below let us fix 
$n\in\mathbb{N}$, finite sets $\mathcal{M}$ and $\mathcal{X}$, finite quantum systems $H$ and $H'$,
and a classical-quantum wiretap channel $(W,{V})$ from $\mathcal{X}$ to $H$ and $H'$.
Just like for the notation $PV$ for a channel $V:\mathcal{X}\to \mathcal{S}(H)$ 
and a probability distribution $P$ over $\mathcal{X}$,
we define $EV$ for a classical channel $E:\mathcal{M} \to \mathcal{P}(\mathcal{X})$, $m\mapsto E_m \equiv E(\cdot|m)$ as
\begin{alignat*}{2}
    EV: &\mathop{}\mathcal{M} &&\to \mathcal{S}(H),\\
        &\mathop{}m           &&\mapsto \sum_{x\in\mathcal{X} } V(x)E(x|m).
\end{alignat*}
Notice that now, for a measurement operator $D$, we have two ways of writing the output probability given $m$, namely $E_m V(D)$ and $EV(D|m)$.

\begin{definition}[Code, error, leakage]
\label{def:code}
\renewcommand{\emph}{\textbf}
\hfill\\
An $(n, |\mathcal{M}|)$  \emph{code}  for $(W,{V})$
is a finite set  
\begin{equation*}
    \mathcal{C}=\qty{E_m,D_m:m\in\mathcal{M}},
\end{equation*} 
where the stochastic encoder $E:\mathcal M\to\mathcal P(\mathcal X^n)$ is a classical channel, and the decoder operators $\qty{D_m: m\in \mathcal{M}}$ form a sub-POVM  on ${H}^{\otimes n}$.
We assume that the POVM is completed by associating the measurement operator $\mathrm{id}_H - \sum_{m\in\mathcal{M}}D_m$ to the the error/abortion symbol of the decoder.

The (maximum) \emph{error} (probability) of $\mathcal{C}$ is defined as 
\begin{align*}
    e(\mathcal{C}, n) 
      &\coloneqq 
    \sup_{m\in\mathcal{M}} E_m W^{\otimes n}(D_m^c)
    \\&=
    \sup_{m\in\mathcal{M}}
    \sum_{x^n \in \mathcal{X}^n}
    E(x^n|m)(1 - W^{\otimes n}(D_m|x^n)),
\end{align*}
where $D^c_m \coloneqq \mathrm{id}_{{H}^{\otimes n}} - D_m$.
For any random variable $M$ over the messages $\mathcal{M}$, 
the \emph{leakage} of $\mathcal C$ with respect to $M$ is defined as
\begin{gather*}
\label{eq:leakage}
\chi\qty(M;EV^{\otimes n}).
\end{gather*}
\end{definition}
Notice that we have chosen maximum rather than average transmission error; the reason is that as we allow the probability distribution to be arbitrary, the correctness and not just the secrecy of the message should be also guaranteed independently of the input distribution. Maximum error is the counterpart to measuring the leakage via maximum indistiguishability between the messages (generally via trace distance~\cite[Section~II-E]{Ha15}, or statistical distance for classical channels).

Observe that we consider codes with stochastic encoders as opposed to deterministic codes. In a deterministic code the encoder is deterministic, namely in \Cref{def:code} instead of a family of probability distributions $\{E_m\}_{m\in\mathcal M}$, the encoder consists of a family of $n$-length strings of symbols $\qty{c_m}_{m\in\mathcal{M}} \subseteq \mathcal{X}^n$.The deterministic encoder can be obtained as a special case of the stochastic encoder by imposing that every probability distribution $E_m$ is deterministic. For message transmission over an ordinary classical-quantum channel, and even for the most general case of  robust message transmission over an arbitrarily-varying  classical-quantum channel, it is enough to use deterministic encoders~\cite{Ahl/Bj/Bo/No,Bo/No}. However, for secret message transmission over wiretap channels, we need to use stochastic encoders~\cite{Ca/Wi/Ye,De}.
\label{detvsran}

Now we will define the coding scheme where both the sender and the receiver have access to common randomness. We do not require this common randomness to be secure against eavesdropping. Effectively, the common randomness simply decides which among a set of classical-quantum codes from \Cref{def:code},will be used.

\begin{definition}[Common-randomness code, error and leakage]
\label{abnjaqcftavc}
\label{def:code-cr}
\renewcommand{\emph}{\textbf}
\hfill\\
An $(n, |\mathcal{S}|, |\mathcal{M}|)$ \emph{common-randomness code} for $(W,{V})$ is a finite subset 
\begin{equation*}
    \qty{\mathcal{C}^{s}=\qty{(E^{s}_m,D^{s}_m):m\in\mathcal{M}}:s\in\mathcal{S}}
\end{equation*}
of the set of $(n, |\mathcal{M}|)$ codes from \Cref{def:code}, 
labeled by a finite set $\mathcal{S}$.

Let $S$, the seed, be a uniform random variable over $\mathcal{S}$. 
The (expected) \emph{error} of $\qty{\mathcal{C}^s:s\in\mathcal{S}}$ is defined as 
\begin{equation*}
    \mathbb{E}_S e(\mathcal{C}^S, n) 
    = \frac{1}{|\mathcal{S}|} \sum_{s\in\mathcal{S}} e(\mathcal{C}^s, n).
\end{equation*}
For any random variable $M$ over the messages independent of $S$, 
the \emph{leakage} of $\qty{\mathcal{C}^{s}}$ with respect to $M$ is defined as
\begin{gather}
\label{eq:leakage-cr}
\chi\qty(M;S,E^SV^{\otimes n}).
\end{gather}
\end{definition}

The definition of leakage reflects the possibility that 
the common randomness is known perfectly to the eavesdropper,
as the leakage is also computed against the common randomness.
Observe, that due to the independence of $S$ and $M$,
the leakage can be written as conditional mutual (or Holevo) information 
between the message $M$ and the output conditioned on the seed $S$:
\begin{equation}
\label{eq:leakage-conditioning}
\chi(M;S,E^SV^{\otimes n})
=\mathbb E_S\chi(M;E^SV^{\otimes n})
=\frac{1}{\mathcal{|S|}}\sum_{s\in\mathcal{S}} \chi(M;E^sV^{\otimes n}).
\end{equation}

The random seed should not be confused with the randomness of the stochastic encoder. In the stochastic encoder, only the sender, but not the receiver, 
randomly chooses a code word to encode a message $m$ according to the probability distribution $E_m$. In the subsequent definitions of achievable rates, the receiver should be able to  decode $m$ even when he only knows $E_m$, but not which code word is actually chosen by the sender. In contrast, a randomly chosen seed $s$ determines  a stochastic encoder $E^{s}$ for the sender and a set of decoder operators $\{D_m^{s}:m\in\mathcal{M}\}$ for the receiver. Correctness is required only for the case that $s$ is known to both the sender and the receiver and that they use the encoder and decoder prescribed by $s$.

\subsection{Capacities}

Next we define the strong and semantic secrecy rates which can be achieved by the codes introduced in the previous subsection. A good code reliably conveys private information to the intended receiver such that the wiretapper's knowledge of the transmitted information can be kept arbitrarily small in terms of the corresponding secrecy criterion. 

\begin{definition}[Strong secrecy]
\label{defofrate}
\renewcommand{\emph}{\textbf}
\hfill\\
A code $\mathcal{C}= \qty{(E_m, D_m) : m\in\mathcal{M}}$ 
is an $(n, R,\epsilon)$ \emph{strong secrecy code} for $(W,{V})$ if
\begin{gather} 
{\log |\mathcal{M}|} \geq nR
\\
e(\mathcal{C},n) < \epsilon\text{ ,}
\label{annian1}
\\
\chi\left(U;EV^{\otimes n}\right) < \epsilon\text{,}
\label{b40}
\end{gather}
where $U$ is the uniform distribution on $\mathcal{M}$.

$R$ is an \emph{achievable strong secrecy rate}
if for every $\epsilon>0$ and sufficiently large $n$,
there exists an $(n, R-\epsilon,\epsilon)$ strong secrecy code.
The \emph{strong secrecy capacity} $C_\textnormal{strong}(W,V)$ 
is the supremum of all achievable strong secrecy rates of $(W,{V})$.
\end{definition}

\begin{definition}[Common-randomness strong secrecy]%
\renewcommand{\emph}{\textbf}%
\hfill\\
A common-randomness code 
$\qty{\mathcal{C}^{s}=\qty{(E^{s}_m,D^{s}_m):m\in\mathcal{M}}:s\in\mathcal{S}}$ 
is an $(n, R,\epsilon)$ \emph{common-randomness strong secrecy code} for $(W,{V})$ if
\begin{gather} 
{\log |\mathcal{M}|} \geq nR\text{,}
\\
\mathbb{E}_S e(\mathcal{C}^S,n) < \epsilon\text{,}
\\
\label{eq:stronsecrecy-cr}
\mathbb{E}_S \chi\left(U;S,E^SV^{\otimes n}\right) < \epsilon\text{,}
\end{gather}
where $U$ is the uniform distribution on $\mathcal{M}$.

$R$ is an \emph{achievable common-randomness strong secrecy rate}
if for every $\epsilon>0$ and sufficiently large $n$,
there exists a $(n, R-\epsilon,\epsilon)$ common-randomness strong secrecy code.
The \emph{common-randomness strong secrecy capacity} $C_\textnormal{strong}(W,V; cr)$ 
is the least upper bound of all achievable common-randomness strong secrecy rates of $(W,{V})$.
\end{definition}

Since codes without common randomness are just a special case of common-randomness codes, we have by construction that 
\begin{equation*}
    C_\textnormal{strong}(W,V)\leq C_\textnormal{strong}(W,V;cr).
\end{equation*}

Strong secrecy, i.e., the requirements of \Cref{b40,eq:stronsecrecy-cr}, is the secrecy criterion which has been used mostly in information-theoretic security until the introduction of semantic security in~\cite{Be/Te/Va}. It provides secrecy if the message random variable is uniformly distributed. Inspired by cryptography,~\cite{Be/Te/Va}~introduced semantic security, where the eavesdropper shall not obtain any information regardless of the probability distribution of the message. This also is the reason why we use the maximum instead of the average transmission error. Semantic security and indistinguishability for the classical-quantum channels were first considered in~\cite{Ha15}.
Here we state the semantic security definitions.

\begin{definition}[Semantic secrecy]
\label{defofratesem}
\label{def:semsec}
\renewcommand{\emph}{\textbf}
\hfill\\
A code $\mathcal{C}= \qty{(E_m, D_m) : m\in\mathcal{M}}$ 
is an $(n, R,\epsilon)$ \emph{semantic secrecy code} for $(W,{V})$ if
\begin{gather}
{\log \mathcal{M}}\geq nR
\\
e(\mathcal{C},n) < \epsilon\text{ ,}
\label{annian1sem}
\\
\max_M \chi\left(M;EV^{\otimes n}\right) < \epsilon,
\label{b40sem}
\end{gather} 
where $M$ is any random variable over the messages $\mathcal{M}$.

$R$ is an \textbf{achievable semantic secrecy rate}
if for every $\epsilon>0$ and sufficiently large $n$,
there exists a $(n, R-\epsilon,\epsilon)$ semantic secrecy code.
The \emph{semantic secrecy capacity} $C_\textnormal{sem}(W,V)$ 
is the supremum of all achievable  semantic secrecy rates of $(W,V)$.
\end{definition}

\begin{definition}[Common-randomness semantic secrecy]
\label{csmcr}
\label{def:semsec-cr}
\renewcommand{\emph}{\textbf}
\hfill\\
A common-randomness code 
$\qty{\mathcal{C}^{s}=\qty{(E^{s}_m,D^{s}_m):m\in\mathcal{M}}:s\in\mathcal{S}}$ 
is an $(n, R,\epsilon)$ \emph{common-randomness semantic secrecy code} for $(W,{V})$ if
\begin{gather}
{\log |\mathcal{M}|} \geq nR
\\
\mathbb{E}_S e(\mathcal{C}^{S},n) < \epsilon\text{ ,}
\\
\label{eq:seccrit_cr}
\max_{M}\chi(M; S, E^{S}V^{\otimes n}) < \epsilon\text{ ,}
\end{gather}
{where} $M$ is any random variable over the messages $\mathcal{M}$.

$R$ is an \textbf{achievable common-randomness semantic secrecy rate}
if for every $\epsilon>0$ and sufficiently large $n$,
there exists a $(n, R-\epsilon,\epsilon)$ common-randomness semantic secrecy code.
The \emph{common-randomness semantic secrecy capacity} $C_\textnormal{sem}(W,V; cr)$ is the supremum of all achievable common-randomness semantic secrecy rates of $(W,{V})$.
\end{definition}

Just like for strong secrecy, and
due to common-randomness codes being more general, 
we have by construction
\begin{equation*}
    C_\textnormal{sem}(W,V)\leq C_\textnormal{sem}(W,V;cr).
\end{equation*}
Similarly, the semantic secrecy condition is stronger, 
meaning that any semantically secure capacity achieving code family,
is also a strongly secure code family and thus
\begin{align*}
    C_\textnormal{sem}(W,V)   &\leq C_\textnormal{strong}(W,V),
    \\
    C_\textnormal{sem}(W,V;cr)&\leq C_\textnormal{strong}(W,V;cr).
\end{align*}
Since the maxima in \Cref{b40sem} and \Cref{eq:seccrit_cr} range over all possible message distributions, semantic security in particular implies message indistinguishability. This means that even if the message random variable can only assume one of two possible values known to the eavesdropper, the eavesdropper cannot distinguish between these two messages. This is not implied by strong secrecy alone.

Notice that since the leakage of the common-randomness codes in \Cref{eq:leakage-cr} is computed against the state at the wiretap and the seed,
bounding the leakage in the common-randomness capacities implies bounding the information about the key carried by the seed.
Thus the common randomness is not required to be secure against eavesdropping,
since the \Cref{eq:stronsecrecy-cr,eq:seccrit_cr} impose that the seed carries no information, 
and thus it is considered to be public.

\subsection{Derandomization}
Derandomization is a standard and widely used technique in information theory, already used by Ahlswede in~\cite{Ahl1}.
As a final result in this section, we apply the derandomization technique to good common-randomness semantic-security codes, namely we construct a semantic-security code without common randomness using a transmission code and a common-randomness semantic-security code with appropriate error scaling. These derandomized codes will essentially be able to produce the common randomness needed to run the common-randomness codes using an asymptotically small number of copies of the channel.
The proof mimics the classical case showed in~\cite{Wi/Bo}.

A simple idea that uses too many channels to generate the seed, however, is to alternate transmission codes and common-randomness semantic-secrecy codes, use the transmission code to generate the seed,
and use it only once in the common randomness semantic-security code.
Depending on the size of the required seed, this may result in too many channels used just for the seed. 
The solution is to simply reuse the seed, thus reducing the total size of $|\mathcal{S}|$ by sharing the same $s\in \mathcal{S}$ for $N$ common-randomness codes.
We thus need to build $(N+1)$-tuple of codewords as the new codewords. Each tuple is a composition of a first codeword that generates the common-randomness and $N$ common randomness-assisted codewords to transmit the messages to the intended receiver.
We start by defining such codes.

\begin{definition}[Derandomizing codes]
\label{def:derandomized}
Let $n,n',N\in\mathbb{N}$. Let 
\begin{itemize}
    \item $\qty{E'_s,D'_s}_{s\in \mathcal{S}}$ 
    be an $(n',|\mathcal{S}|)$ code,
    \item $\qty{\qty{E^s_m, D^s_m \!:m\in\!\mathcal{M}}\!\!: s\in\mathcal{S}}$
    be an $(n, |\mathcal{S}|, |\mathcal{M}|)$ common-randomness code,
\end{itemize}
and define $\bar{\mathcal{M}} \coloneqq \mathcal{M}^N$ and for any $\bar{m}\in\bar{\mathcal{M}}$
\begin{align*}
    E^s_{\bar{m}} &\coloneqq E^s_{m_1} \cdot \ldots\cdot E^s_{m_N}
    &
    D^s_{\bar{m}} &\coloneqq D^s_{m_1} \otimes \dots \otimes D^s_{m_N}.
\end{align*}
We define their $(n'+nN, |\mathcal{M}|^N)$ derandomized code $\bar{\mathcal{C}}$
to be the code (without common randomness), such that for any message $\bar{m}\in\bar{\mathcal{M}}$
\begin{itemize}
\item 
    the encoder samples from a uniform seed $S$ and then, conditioned on the values $s$,
    uses the Kronecker product encoder $E'_s \cdot  E^s_{\bar{m}}$, thus
    \begin{align*}
    \bar{E}_{\bar{m}} \coloneqq \mathbb{E}_S  E'_S \cdot E^S_{\bar{m}}
    \text{, and} 
    \end{align*}
\item 
    the decoder for the message is the coarse graining of decoders over $s$ 
    \begin{align*}
    \bar{D}_{\bar m} 
    &\coloneqq
    \sum_{s\in\mathcal{S}}  D'_s \otimes D^s_{\bar{m}}
    .
    \end{align*}
\end{itemize}
\end{definition}

Note that the random seed in the derandomizing code becomes part of the stochastic encoding process of the code. As we expect, the error and the leakage of the derandomizing code is not worse than the sum of the errors and leakage of all the codes used in the process. This can be easily proved by simply applying the standard techniques (cf.~\cite{Ahl/Bli,Bo/Ca/De}) for derandomization with uniform distributed inputs on derandomization with arbitrary distributed inputs. Notice that the standard proof of security (cf.~\cite{Bo/Ca/De}) is nothing more than applying the quantum data processing inequality (cf.~\cite{Wil}) when we consider the derandomizing  code  as a function of its first part. Thus this argument  works for any inputs distribution. Nevertheless  we give a proof for the sake of completeness.

\begin{lemma}
\label{thm:Nderandom}
Let $\mathcal{C'}$ be an $(n',\frac{1}{n'}\log|\mathcal{S}|,\epsilon')$ transmission code,
and let $\qty{\mathcal{C}^s}_{s\in\mathcal{S}}$ be an $(n, R, \epsilon)$
common-randomness semantic-secrecy code.
Let $\bar{n} \coloneqq n'+nN$, then the $N$-derandomized code $\bar{\mathcal{C}}$ 
is an $(\bar{n}, \frac{n N}{\bar{n}} R, \epsilon'+\epsilon N)$ semantic-secrecy code.
\end{lemma}

\begin{proof}
The $N$-derandomized code has size $|\mathcal{M}|^N$,
thus the rate is $N \log |\mathcal{M}| / \bar{n} \geq n N R / \bar{n}$.
We just need to bound the error and leakage of the new code.

For the error of $\bar{\mathcal{C}}$, by standard argument we have
that for every $\bar{m}\in\bar{\mathcal{M}}$
\begin{align*}
    \bar{E}_{\bar{m}} &W^{\otimes \bar{n}}(\bar{D}_{\bar{m}}^c)
    \\& =
    1 - \sum _s \frac{1}{|\mathcal{S}|} E'_{s}W^{\otimes n'} 
    \otimes E^s_{m_1}W^{\otimes n} \otimes \dots 
    \otimes E^s_{m_N}W^{\otimes n} (\bar{D}_{\bar{m}}) 
    \\& = 
    1 - \sum _{s,s'} \frac{1}{|\mathcal{S}|} E'_{s}W^{\otimes n'}(D'_{s'}) 
    \otimes E^s_{m_1}W^{\otimes n}(D^{s'}_{m_1}) \otimes \dots 
    \otimes E^s_{m_N}W^{\otimes n}(D^{s'}_{m_N}) 
    \\& \leq
    1 - \sum _{s} \frac{1}{|\mathcal{S}|} E'_{s}W^{\otimes n'}(D'_{s}) 
    \otimes E^s_{m_1}W^{\otimes n}(D^{s}_{m_1}) \otimes \dots 
    \otimes E^s_{m_N}W^{\otimes n}(D^{s}_{m_N}) 
    \\& \leq
    1 - \sum _{s} \frac{1}{|\mathcal{S}|} (1-\epsilon') (1-\epsilon)^N
    \\& \leq
    \epsilon' + \epsilon N,
\end{align*}
and thus $e(\bar{\mathcal{C}},\bar{n}) \leq \epsilon' + \epsilon N$.

For the leakage, since we are reusing the seed and the seed is shared via the transmission code, a uniform seed does not map a random message to a uniformly random input to the channel. Thus we need to reduce the security to the security of the single codes.
Recalling that the the Holevo information is actually mutual information, by data processing and \Cref{eq:leakage-conditioning}, we have
\begin{align*}
    \chi\qty(\bar{M}; \bar{E} V^{\otimes \bar{n}})
      &=
    \chi\qty(\bar{M}; E'_SV^{\otimes n'}\otimes E^S_{\bar{M}}V^{\otimes nN})
    \\&\leq
    \chi\qty(\bar{M}; S, E^S_{\bar{M}}V^{\otimes nN})
    \\&\leq
    \mathbb{E}_S \chi\qty(\bar{M}; E^S_{\bar{M}}V^{\otimes nN})
    \\&= 
    \mathbb{E}_S \qty
    [H\qty(\mathbb{E}_{\bar{M}}       E^S_{\bar M}V^{\otimes nN}) - 
           \mathbb{E}_{\bar{M}} H\qty(E^S_{\bar M}V^{\otimes nN})].
\intertext{Since the message encoder is a Kronecker product of encoders, 
and the channel is memoryless, we have 
$H\qty(E^S_{\bar M}V^{\otimes nN})= \sum H\qty(E^S_{M_i} V^{\otimes n})$, and thus}
    \chi(\bar{M}; \bar{E} V^{\otimes \bar{n}})
      &\leq 
    \mathbb{E}_S \qty
    [H\qty(\mathbb{E}_{\bar{M}}                      E^S_{\bar{M}}V^{\otimes nN}) - 
           \mathbb{E}_{\bar{M}} \sum_{i=1,\dots,N} H(E^S_{M_i}    V^{\otimes n })].
\intertext
{This together with $H(XY) \leq H(X) + H(Y)$ applied to 
$H\qty(\mathbb{E}_{\bar{M}}   E^S_{\bar M}V^{\otimes nN})$ gives}
    \chi(\bar{M}; \bar{E} V^{\otimes \bar{n}})
      &\leq 
    \mathbb{E}_S \qty
    [\sum_{i=1,\dots,N} H\qty(\mathbb{E}_{M_i}           E^S_{M_i}V^{\otimes n}) - 
     \sum_{i=1,\dots,N}       \mathbb{E}_{\bar{M}} H\qty(E^S_{M_i}V^{\otimes n})]
    \\&=
    \mathbb{E}_S \sum_{i=1,\dots,N} 
    \chi\qty({M}_i; E^S_{M_i}V^{\otimes n})
    \\&=
    \sum_{i=1,\dots,N} 
    \chi\qty({M}_i; S, E^S_{M_i}V^{\otimes n})
    \\&\leq 
    \epsilon N \leq \epsilon'+ \epsilon N,
    \yesnumber
    \label{thm:Nderandomineq}
\end{align*}
and the proof is concluded.
\end{proof}

Notice that the argument works for any distribution of $\bar{M}$, 
the single uses of the semantic-secrecy code do not need to have independent messages.
This is usually a point of difference with the derandomization techniques used for strong secrecy. In strong secrecy $\bar{M}$ is only required to be uniformly distributed, which makes each $M_i$ already independent and also uniformly distributed. This allows for an easier but not fully general argument, since the leakage of the derandomized code is actually equal to the sum of the leakages of the single internal codes.

We will use the above in \Cref{sec:bri} to derandomize the explicit constructions of semantic secrecy codes.

\subsection{Quantum  Channels}
\label{totcwtraqc}
\label{sec:quantum}
\renewcommand{\mathfrak}{\mathrm}

The results from classical secret message transmission over classical-quantum channels can usually be carried over to fully quantum channels. Moreover, this is optimal, in the sense that it is usually enough to just prepend a classical-quantum preprocessing channel to many copies of the quantum channel and then use the coding for the resulting classical-quantum channel.
The extension to quantum channels reduces to simply proving \Cref{NbdettctiVw}, which is straightforward and uses quite general arguments.
More precisely, since the encoding of classical messages for any quantum channel will need to map the classical messages to quantum states, the resulting effect at the sender is again a classical-quantum channel, and thus we can reduce the analysis to what we have done so far for classical-quantum channels.

For classical and classical-quantum channels, the wiretap channel must be given, in the sense that an assumption must be made about the output seen at the eavesdropper, simply because the worst case scenario, that the eavesdropper receives a noiseless copy of the input, is always physically possible.
This is not the case for quantum channels, where one of the aspects of no-cloning implies that a copy of the input quantum state cannot be made, and the worst case interaction with the environment can be deduced from the noise in the channel.
Since there is a limit to the information that it is leaked to the environment, there is thus also a limit to the information of the eavesdropper, and we can then remove any assumption in that respect, and identify the eavesdropper with the environment~\cite{Be/Br,De}.

\smallskip
Let now $\mathfrak{P}$ and $\mathfrak{Q}$ be quantum systems, and let $W$ be a quantum  channel. 
We assume, as usual in the quantum setting, the worst case scenario, namely that the environment $\mathfrak{E}$ is completely under the control of the eavesdropper, which is in contrast with the classical and classical-quantum setting where this worst case scenario does not allow for secrecy.
This automatically defines the wiretap channel $(W,V)$ for any given quanutm channel $W$ to the intended receiver.
However the results below work in general for any allowed pair of quantum channels $(W,V)$ on the same input.



\begin{definition}
A quantum wiretap channel from a sender $\mathfrak{P}$ to a receiver $\mathfrak{Q}$ with eavesdropper $\mathfrak{E}$ 
is a pair of complementary channels $(W,{V})$,  where
$W: \mathcal{S}(H_{\mathfrak{P}}) \to \mathcal{S}(H_{\mathfrak{Q}})$ and ${V}: \mathcal{S}(H_{\mathfrak{P}}) \to \mathcal{S}(H_{\mathfrak{E}})$ are defined as
$W(\rho) = \tr_{\mathfrak{E}} 
\left(U \rho U^*\right)$ and
$V(\rho) = \tr_{\mathfrak{Q}} 
\left(U \rho U^*\right)$
for some isometry $U: H_{\mathfrak{P}} \to H_{\mathfrak{Q}} \otimes H_{\mathfrak{E}}$.

\end{definition}

\begin{remark}
Without the assumption that the eavesdropper might have full access to the environment, the treatment of the semantic secrecy capacity is still the same. 
In this case the wiretap channel must be specified explicitly as $(W,V)$, where both $W$ and $V$ are quantum channels. However not all pairs are allowed, as $V$ must be a channel that can be recovered from the environment. The generalization is that $W$ and $V$ must be of the form $W = \tr_{\mathfrak{E}\mathfrak{R}} \qty(U \rho U^*)$ and $V = \tr_{\mathfrak{Q}\mathfrak{R}} \qty(U \rho U^*)$, where now the isometry $U: H_{\mathfrak{P}} \to H_{\mathfrak{Q}}\otimes H_{\mathfrak{E}}\otimes H_{\mathfrak{R}}$ maps to three systems, the intended receiver, the eavesdropper, and an environment not in possession of the eavesdropper.
\end{remark}

We can transmit both classical and quantum information over quantum channels. For the transmission of classical information via a quantum channel, we  first have to convert a classical message into a quantum state. We assume that the states produced in the input system are constructed depending on the value of $x\in\mathcal{X}$,  where $\mathcal{X}$ is a finite set of letters.
Let thus $F:\mathcal{X} \to \mathcal{S}(H_{\mathfrak{P}})$ be this classical-quantum channel. The composition with a quantum channel $W$ defines the classical-quantum channel ${W} \circ F:\mathcal{X}\to\mathcal{S}(H^\mathfrak{Q})$; to keep a consistent notation we define $F W \equiv W \circ F$.
With this notation, the definitions present only minimal changes in comparison to the classical-quantum wiretap channels above.
A code for the quantum channels now simply needs to input quantum states instead of classical values.

\begin{definition}
\renewcommand{\emph}{\textbf}
An $(n, |\mathcal{M}|)$ \emph{quantum code} for a quantum channel $W$ consists of a finite set $\mathcal{C}= \qty{E_m,D_m:m\in\mathcal{M}}$, where the stochastic encoder $E:\mathcal{M}\to \mathcal{S}(H_{\mathfrak{P}} ^{ \otimes n})$ is a classical-quantum channel,
and the decoders $\qty{D_m: m\in \mathcal{M}}$ form a sub-POVM.

The  error of $\mathcal{C}$ is defined as
\[e(\mathcal{C}, n) \coloneqq \frac{1}{|\mathcal{M}|} \max_{m\in\mathcal{M}} E_m W^{\otimes n}(D_m^c)\text{ .}\] 

The leakage of a message random variable $M$ over $\mathcal{M}$ is defined as 
\[\chi\qty(M;E_m V^{\otimes n}),\]
where $V$ is the complementary channel to the environment.
\end{definition} 

The rates and capacities can then be defined exactly as is done for classical-quantum channels.
Since we will use these definitions only briefly in \Cref{NbdettctiVw}, we limit ourselves to directly defining the capacities.

\begin{definition}
\label{defofrateqq}
The strong secrecy capacity $C_\textnormal{strong}(W)$ is the largest real number such that for every $\epsilon>0$ and sufficiently large $n$ there exists a finite set $\mathcal{X}$ and an  $(n, |\mathcal{M}|)$ code $\mathcal{C} =  \qty{E_m, D_m : m\in\mathcal{M}}$,  such that 
\begin{gather} 
{\log |\mathcal{M}|}> n(C_\textnormal{strong}(W)-\epsilon)
\\
\label{annian1qq}
 e(\mathcal{C},n) < \epsilon\text{ ,}
\\
\label{b40qq}
\chi\left(U;E_m V^{\otimes n}\right) < \epsilon\text{,}
\end{gather}
where $U$ is the uniform random variable over $\mathcal{M}$.
\end{definition}

\begin{definition}
\label{defofrateqqsem}
The  semantic   secrecy
capacity of $W$, denoted by $C_\textnormal{sem}(W)$
is the largest real number such that for every $\epsilon>0$ and sufficiently large $n$ there exists an  $(n, |\mathcal{M}|)$ code $\mathcal{C} = \{E_m, D_m : m\in\mathcal{M}\}$, such that  for any  random variable $M$ with  arbitrary distribution  on $\mathcal{M}$
\begin{align} 
{\log |\mathcal{M}|}> n(R-\epsilon)
\\
\label{annian1qqsem}
e(\mathcal{C},n) < \epsilon\text{ ,}
\\
\label{b40qqsem}
\chi\left(M;E_m V^{\otimes n}\right) < \epsilon\text{ .}
\end{align} 
\end{definition}

Notice that the choice of environment channel does not affect the definitions of capacity.
Let ${V}$ and $V'$ be two distinct complementary channels to $W$, then ${V}'$ and ${V}$ are equivalent in the sense that there is a partial isometry
$U$ such that for all input states $\rho \in  \mathcal{S}(H_{\mathfrak{P}})$ we have ${V}'(\rho)= U^{*} {V}(\rho) U$~\cite{Pa,Ho2}.
The action of the partial isometry is reversible and thus the leakage is the same (being a mutual information, which is non-increasing under local operations).
Therefore the security criteria in \Cref{defofrateqq,defofrateqqsem} does not depend on the choice of complementary channel.

\bigskip
With the definitions in place, the following sections are dedicated to prove our results.
In \Cref{sec:non-constructive} we prove that we can change any strong secrecy capacity achieving codes into semantic secrecy capacity achieving codes.
However the result is non constructive, which is why in \Cref{sec:bri} we provide a semi-constructive proof where we concatenate functions to suitable transmission codes to convert them into semantic secrecy capacity achieving codes.
The final section is dedicated to generalizing the results from classical-quantum channels to quantum channels.

\section{Semantic Secrecy Capacity}
\label{sscocqwc}
\label{sec:non-constructive}

\begin{figure}
    \centering
    \includegraphics{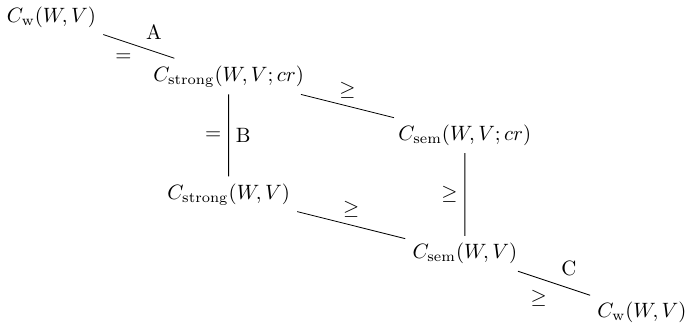}
    \caption{Relationship between the secrecy capacities
    and the coding theorems.
    ``A'' follows from \cite{De,Bo/Ca/De2}, ``B'' from \cite{De,Bo/Ca/De2} and ``C'' from \Cref{thm:capacity}.
    The inequalities without reference are obvious and follow from the definitions;
    allowing common randomness can only increase the capacity.
    Similarly, relaxing from semantically secure code 
    to strongly secure codes can only increase the capacity.
    It follows from $C_{\text{w}}(W,V)$ being both the upper bound and the lower bound,
    that all the quantities are equal.}
    \label{fig:capacities}
\end{figure}

We denote
\begin{equation}
\label{eq:wiretap-capacity}
C_{\text{w}} (W,V) 
\coloneqq \sup_{n\to\infty}\frac{1}{n} \max_{P,E}
(\chi(P;EW^{\otimes n})- \chi(P;EV^{\otimes  n})),
\end{equation}
where "w" stands for wiretap, and
the maximum is taken over finite input sets $\mathcal{M}$,
input probability distributions $P$ on $\mathcal{M}$,
and classical channels $E:\mathcal{M}\to P(\mathcal{X}^n)$.

$C_{\textnormal{w}} (W,V)$ was first proven in~\cite{De} to equal the strong secrecy capacity of the classical-quantum wiretap channel.
The result was extended in~\cite{Bo/Ca/De2} to the common-randomness strong secrecy capacity as a particular case of arbitrarily-varying classical-quantum wiretap channels.
Namely, we have 
\begin{equation}
\label{lem:strong_sec_CQ_wiretap}
    C_\textnormal{strong} (W,V; cr) = C_\textnormal{strong} (W,V) = C_{\textnormal{w}} (W,V).
\end{equation}
For now we will not actually use the explicit expression of $C_\textnormal{w}$.
Since a semantically secure code is always also strongly secure, the converse theorems for strong secrecy are also strong converses for semantic secrecy, as displayed schematically in \Cref{fig:capacities}.

By the results of \cite{Ha15}, it is easy to see that
$C_\textnormal{sem} (W,V) \geq C_\textnormal{strong} (W,V)$,
when we apply the standard expurgation technique  given in \cite{Ha15} to 
our channel to
convert a strong secrecy code into a semantic secrecy code without asymptotic rate loss (see  also \cite{Hay/Mat} for the classical finite  case):
For any $\epsilon>0$, by the  definition of $C_{\text{strong}}(W,V)$,
there exists a $\delta>0$ such that for all sufficiently large $n$, there exists an $(n,|\mathcal{M}_n|)$ strong secrecy code $\mathcal{C}_n$ satisfying
$|\mathcal{M}_n| \geq 2^{n(C_\textnormal{strong} (W,V)-\epsilon)}$,
$e(\mathcal{C}_n, n)\leq2^{-n\delta}$, and
   $ \chi(U_n; EV^{\otimes n})\leq2^{-n\delta}$,
where $U_n$ is the uniform distribution over $\mathcal{M}_n$, $E$ is the encoder of the code. 
We have $\chi(U_n; EV^{\otimes n}) = \frac{1}{|\mathcal{M}_n|}\sum_{m\in \mathcal{M}_n} D(E_m V^{\otimes n}\|U_n E V^{\otimes n})$. 
Thus, as per expurgation in \cite{Ha15},  there is a subcode $\mathcal{C}_n'$ of size $\mathcal{M}_n'\geq \mathcal{M}_n/2$
such that we can choose a $0<\delta'<\delta$ such that for any $m \in \mathcal{M}_n'$ and sufficiently large $n$,
we have 
$D(E_m V^{\otimes n}\|U_n E V^{\otimes n}) \leq 2 \cdot 2^{-n\delta} < 2^{-n\delta'}$.
Notice that the encoder is the same, we are just restricting the set of messages.
Then, for any probability distribution $P$ on the new message set $\mathcal{M}_n'$, by \cite[Eq.~(4.7)]{Hay}
we have 
\begin{align}
\chi(P ; E V^{\otimes n}) 
  &= \sum_{m\in \mathcal{M}_n} P_m D( E_m V^{\otimes n}\|P E V^{\otimes n}) 
\\&= \min_{\sigma}\sum_{m\in \mathcal{M}_n}P(m) D(E_m V^{\otimes n}\|\sigma) 
\\&\leq \sum_{m\in \mathcal{M}_n} P_m D(E_m V^{\otimes n}\|U_n E V^{\otimes n})
\\&\leq 2^{-n\delta'}.
\end{align}
Since  $\mathcal{C}_n'$  is a subcode we also have $e(\mathcal C_n',n)\leq 2^{-n\delta}< 2^{-n\delta'}$.
Since $\mathcal M_n'$ contains at least half of the messages of $\mathcal M_n$, the expurgation technique given in \cite{Ha15} immediately delivers the following theorem, independently of the type of channel.


\begin{theorem}
\label{thm:capacity}
Let $(W,V)$ be a classical-quantum wiretap channel. With the same notation as in \Cref{lem:strong_sec_CQ_wiretap}, we have
\begin{equation*}
    C_\textnormal{sem} (W,V) = C_\textnormal{sem} (W,V; cr)  
    = C_\textnormal{strong} (W,V) = C_{\textnormal{w}} (W,V).
\end{equation*}  
\end{theorem}

Notice  that the expurgation technique given in \cite{Ha15} 
makes no assumption on the dimension of the quantum systems, namely the outputs of the wiretap channels considered below can be infinite dimensional.

\subsubsection*{Quantum Channels}
Let $W$ now be a quantum channel, thus defining the quantum wiretap channel to be the complementary channel to the environment.
Just like the case of the classical-quantum channel, the strong and semantic secrecy capacities are equal. This time, rather than transforming a strong secrecy code into a semantic secrecy code, we simply generalize the result from classical-quantum channels to quantum channels, in the same way it was done for strong secrecy in~\cite{De}.

In~\cite{De} it was proven that the strong secrecy  capacity $C_\textnormal{strong}(W)$ can be computed using the following multi-letter formula.
\begin{equation} 
C_\textnormal{strong}(W) =
  \sup_{n\in\mathbb{N}} \frac{1}{n} \sup_{\mathcal{X}, F, P}\left(
  \chi(P,F W^{\otimes n}) -\chi (P,F V^{\otimes n})\right),
\end{equation}
where ${V}$ is the channel to the environment
defined by  $W$.
The supremum is taken over all chosen  finite sets $\mathcal{X}$, classical/quantum channels $F:\mathcal{X}\to \mathcal{S}(H_{\mathfrak{P}}^{\otimes n})$, and probability distributions $P$ on $ \mathcal{X}$.
Notice how the classical-quantum channel is allowed to output entangled states between the inputs of the channels.

Just like for classical-quantum channels, 
any semantic secrecy code is also a strong secrecy code, and the strong secrecy capacity is a converse on the semantic secrecy capacity. Again, we only need the achievability proof.
The achievability of this rate follows directly from the achievability of the wiretap capacity for classical-quantum channels. Since the proof is actually independent of the structure of the secrecy criterion, the same proof for strong secrecy also works for semantic secrecy.

\begin{corollary}
\label{NbdettctiVw}
Let $W$ be a quantum channel. We have
\begin{equation} 
C_{sem}(W) =
  \sup_{n\in\mathbb{N}} \frac{1}{n} \sup_{\mathcal{X}, F, P}\left(
  \chi(P,F W^{\otimes n}) -\chi (P,F V^{\otimes n})\right)\text{ .}
\end{equation}
\end{corollary}

\begin{proof} 
We prove the claim for any fixed $n$, $\mathcal{X}$, $P$ and $F$, namely that 
\[\frac{1}{n} \qty
( \chi(P,F W^{\otimes n}) 
- \chi(P,F V^{\otimes n}))\]
is an achievable rate, the supremum then follows automatically.

Notice that $\chi(P,F W^{\otimes n}) - \chi(P,F V^{\otimes n})$ is already an achievable rate for the classical-quantum channel $(F W^{\otimes n},F V^{\otimes n})$, and thus for all $\epsilon>0$ and all $n'$ there exist an $(n', \chi(P,F W^{\otimes n}) - \chi(P,F V^{\otimes n}) -\epsilon,\epsilon)$  code $\qty{E_m,D_m}$ for $(F W^{\otimes n},F V^{\otimes n})$ as proven in the previous sections. It follows by construction and definition that $\qty{E_m F^{\otimes n'}, D_m}$ is a $(n'n, \chi(P,F W^{\otimes n}) - \chi(P,F V^{\otimes n}) -\epsilon,\epsilon)$ code for $W$ with rate divided by $n$.
\end{proof} 

Since we can reduce the classical semantic secrecy capacity of quantum channels to the one of classical-quantum wiretap channels, we can restrict ourselves to the latter in our analysis.

\bigskip
We have proven that whenever strong secrecy is achievable, 
then semantic security is also achievable. 
However, the proof technique does not tell us how to practically construct such codes, 
and the subset of semantically secure messages chosen in \Cref{thm:capacity} 
will in general depend on the channel and the code. 
In the next section we will address this issue and show how to construct such codes,
similar to how hash functions are used to achieve strong secrecy.

\section{Semantic Security with BRI Functions}
\label{sscocqwcBRI}
\label{sec:bri}

However,  the expurgation technique gives us only an existence statement, 
and does not answer the question of how to choose the semantically secure message subsets. 
In this Section we introduce BRI functions and use them to construct semantic secrecy capacity achieving BRI modular codes in \Cref{pdpaoimtcr}, thus also providing an alternative to the achievability proof of \Cref{thm:capacity} in the previous section. 
We will construct such codes requiring common randomness,
and will at first only show achievability via common-randomness BRI modular codes. An additional derandomization step will be  required to construct codes without common randomness.
The idea behind the construction of semantic-secrecy BRI modular codes is similar to the way in which strong secrecy codes are constructed, using first a transmission code to correct all the errors, but substituting the use of strongly universal hash functions with the use of BRI functions to erase the information held by the eavesdropper.
Just like hash functions, BRI functions require a random seed known to the sender and receiver, which is why we provide it as common randomness.
Providing the seed via common randomness makes construction easier and the proof conceptually clear. However, the assumption of common randomness as an additional resource is quite strong. In the end of the  section 
we prove that the random seed can be generated by the sender and be made known to the receiver using the channel without sacrificing capacity, a process known as derandomization.

An approach achieving the semantic secrecy rate using ``standard'' secure codes has been delivered in \cite{Hay_exp}, where the result of \cite{De} was extended. In \cite{Hay_exp}, it has been demonstrated how to obtain a semantic secure code from a ``standard'' secure code by using a hash function,  when the ``standard'' secure code is a linear code and the channel is an additive channel. 
\cite{Hay_exp} and \cite{Ha13} extended this technique to
 deliver an explicit construction of a secrecy transmission code
with strong secrecy for general classical channel. In additional, 
this method  also work for a continuous input alphabet. The
same technique  can be easily extended to  a classical-quantum
code (cf. \cite{Ha14}).
This technique has been also applied  in \cite{Ha15} for additive fully-quantum channels, when the eavesdropper has access to the whole environment.
Together with the expurgation technique (cf. \Cref{sscocqwc}),
these results deliver an extend proof for \Cref{thm:capacity}.

Our results, both for classical-quantum channels and for fully quantum channels with eavesdropper having access to the whole environment, 
are more general, since our techniques 
deliver an explicit construction of a secrecy transmission code
with semantic secrecy using {\bf any code} ensuring strong
security on any wiretap channel.

To our knowledge, we are the first to show that {\bf every} code with or without common randomness which achieves strong secrecy has a subcode with the same asymptotic rate which achieves semantic secrecy measured in terms of Holevo information. Thus, the passing from the strongly secret code to the semantically secret subcode is highly nonconstructive. An important aspect of THIS modular construction using BRI functions is that there is hope that it can be implemented in practice. 

\subsection{BRI functions}

We will define define what biregular irreducible (BRI) functions are in this subsection, 
and prove the key properties that we will use to achieve semantic security.
The properties we will prove are independent of communication problems
like the classical-quantum wiretap channels we consider.
They are simply related to the structure of BRI functions 
and how they are used as an input to classical-quantum channels.
Thus the channels and the input spaces in this subsection
are not to be confused with the actual wiretap channel
and its inputs, as will be made clearer below.

We will be looking at families of functions $f_s(x)$,
namely functions of two inputs $f: \mathcal{S}\times\mathcal{X}\to\mathcal{N}$, 
and at their preimages in $\mathcal{X}$
\begin{equation*}
    f_s^{-1}(m)\coloneqq \qty{x\in\mathcal{X}: f_s(x)=m},
\end{equation*}
where $s\in\mathcal{S}$ will be the seeds in construction of the modular codes.

\begin{definition}[Biregular functions]
Let $\mathcal{S}$, $\mathcal{X}$, $\mathcal{N}$ be finite sets.
A function $f: \mathcal{S}\times\mathcal{X}\to\mathcal{N}$ 
is called biregular if there exists a regularity set $\mathcal{M}\subseteq\mathcal{N}$ such that for every $m\in\mathcal{M}$ 
\begin{enumerate}[label=\alph*)]
    \item \label{bri:sregular}
    $d_{\mathcal{S}}\coloneqq |\{x: f_s(x) = m\}|= |f_s^{-1}(m)|$ 
    is non-zero and independent of $s$;
    \item \label{bri:xregular}
    $d_{\mathcal{X}}\coloneqq|\{s: f_s(x) = m\}|$ 
    is non-zero and independent of $x$.
\end{enumerate}
\end{definition}
For any biregular function $f: \mathcal{S}\times\mathcal{X}\to\mathcal{N}$ and any $m\in\mathcal{M}$
we can define a doubly stochastic matrix $P_{f,m}$~\cite{Wi/Bo} with coefficients defined as
    \begin{align}
    \label{bri:stochastic-matrix}
    P_{f,m}(x,x') &\coloneqq
        \frac{1} {d_{\mathcal{S}}d_{\mathcal{X}}}
        \cdot {|\{s: f_s(x)= f_s(x') =m\}|}.
    \end{align}
In other words, $P_{f,m}(x,x')$ is the normalized number of seeds $s\in\mathcal{S}$ such that both $x$ and $x'$ are in the preimage $f_s^{-1}(m)$.
Since $P_{f,m}$ is stochastic, its largest eigenvalue is $1$ and we define $\lambda_2(f,m)$ to be the second largest singular value of $P_{f,m}$.

\begin{definition}[BRI functions]
\label{afsxnicabri}
Let $\mathcal{S}$, $\mathcal{X}$, $\mathcal{N}$ be finite sets.
A biregular function $f: \mathcal{S}\times\mathcal{X}\to\mathcal{N}$ 
is called irreducible if
    \label{bri:eigenvalue}
    $1$ is a simple eigenvalue of $P_{f,m}$, namely if $\lambda_2(f,m)< 1$, for every $m\in\mathcal{M}$.
\end{definition}
Notice that $d_{\mathcal{S}}$ and $d_{\mathcal{X}}$ might depend on $m$.
However, for the known BRI function construction 
these are indeed a constant parameter~\cite{Wi/Bo}.

Biregularity puts a strong restriction on the behaviour of the function. 
In particular, any $m$ is a possible output of any $s$ or $x$,
with the right $(s,x)$ pair. 
If for a fix $m$ we consider the incidence matrix $I_{sx} = \delta_{m,f_s(x)}$,
which we can think of it as representing the $m$ section of the graph of $f_s(x)$, 
then we can visualize \cref{bri:xregular,bri:sregular} as $I_{sx}$ 
having the same number of $1$'s in each row, and similarly in each column.
For example, ignoring \cref{bri:eigenvalue} and omitting the zeros, 
a possible $I_{sx}$ for a given $m$ might look like 
\begin{gather*}
\mathcal{X}
\\\mathcal{S}
\begin{pmatrix}
1& & &1&1&1& & \\
 &1&1& & &1& &1\\
1&1& &1& &1& & \\
 & &1& &1& &1&1\\
1&1& &1& & &1& \\
 & &1& &1& &1&1
\end{pmatrix}
\end{gather*}
with $d_{\mathcal{S}}=4$ and $d_{\mathcal{X}} = 3$.
An important consequence is the following relation
\begin{equation}
d_{\mathcal{X}}|\mathcal{X}|=d_{\mathcal{S}}|\mathcal{S}|,
\label{bri:sizes}
\end{equation}
easily derived from 
\begin{align*}
|\{(s,x): f_s(x) = m\}| 
&= \sum_{s} |\qty{x: f_s(x) = m}|
 = \sum_s d_{\mathcal{S}}
 = d_{\mathcal{S}} |\mathcal{S}|
\\
&= \sum_{x} |\qty{s: f_s(x) = m}|
 = \sum_x d_{\mathcal{X}}
 = d_{\mathcal{X}} |\mathcal{X}|.
\end{align*}
In~\cite{Wi/Bo} it was shown that for every $d \geq 3$ and $k\in\mathbb{N}$, there exists a BRI function $f:\mathcal{S}\times\mathcal{S}\to\mathcal M$, satisfying $|\mathcal{S}|=2^kd$, $|\mathcal{M}|=2^k$ and 
\begin{equation}
    \lambda_2(f,m) \leq \frac{4}{d}.
\end{equation}
For our constructions, the above is all we need to know, 
we will not need to know how these functions are constructed.

The above are key features that are used to provide security.
BRI functions play the equivalent role of hash functions for strong secrecy,
so we need to show how they can be used to reduce the 
Holevo information at the output of the channel,
which is what we will do in the remainder of the subsection.
This channel must not be confused with the wiretap,
which we will consider only later in the next subsection.
For the rest of this subsection we fix a finite set $\mathcal{X}$,
a quantum system $H$, a classical-quantum channel
\begin{equation*}
    V: \mathcal{X}\to \mathcal{S}(H),
\end{equation*}
where $V$ is not necessarily a wiretap, and the random variables
\begin{align*}
    S&&&\text{uniform random variable over the seed }\mathcal{S}
    \\
    M&&&\text{random variable over the message }\mathcal{M} 
\end{align*}
which are always assumed to be independent.
In the next subsection, $V$ will actually be the composition 
of the encoder of the transmission code with the actual wiretap;
thus $\mathcal{X}$ will be the message space of the transmission code.
The space $\mathcal{M}$ will be the message space of the wiretap code
and $\mathcal{S}$ the space of the common randomness.
Given a seed $s$, the encoding of a message $m$ happens 
by picking an uniformly random element of the preimage $f_s^{-1}(m)$.
The definition of BRI functions and these conditions, then, are such that fixing the message and choosing the seed at random 
produces a uniformly random encoding, as will be explained now more precisely.
For this pupose, with some abuse of notation,
we will allow classical-quantum channels  to take subsets as inputs, 
with the convention that the resulting state
is the uniform mixture over the outputs of the elements in the set.
Namely, for $\mathcal{D}\subseteq \mathcal {X}$ we will define
\begin{align*}
V(\mathcal{D}) &\coloneqq \frac{1}{\mathcal{|D|}} \sum_{x\in\mathcal{D}} V(x),
\intertext{in particular, the above defines }
V\circ f_s^{-1}(m) &= \frac{1}{d_{\mathcal{S}}}\sum_{x\in f_s^{-1}(m)}V(x),\\
V(\mathcal{X}) &= \frac{1}{\mathcal{|X|}} \sum_{x\in\mathcal{X}} V(x),
\end{align*}
which we will repeatedly use.
In particular, from \Cref{bri:sizes}, stating that 
$d_{\mathcal{X}}|\mathcal{X}|=d_{\mathcal{S}}|\mathcal{S}|$,
it follows immediately that 
\begin{equation}
     P[f_S^{-1}(m)=x]
    =\frac{1}{\lvert\mathcal{S}\rvert}\sum_s\frac{1}{d_{\mathcal{S}}}1_{\{f_s(x)=m\}}
    =\frac{d_{\mathcal{X}}}{d_{\mathcal{S}}\lvert\mathcal{S}\rvert}
    =\frac{1}{\lvert\mathcal{X}\rvert},
\end{equation}
and thus for any $m\in\mathcal{M}$
\begin{equation}
    \label{eq:seed-uniformx}
    \mathbb{E}_S V\circ f_S^{-1}(m) = V(\mathcal{X}),
\end{equation}
which means that not knowing the seed makes the output independent of the message.


\smallskip
We can now start bounding the information at the output of the channel $V$, 
and ultimately will need to be able to show the semantic secrecy conditions.
As said, we focus for now on the common-randomness semantic security,
\Cref{eq:seccrit_cr} of \Cref{def:semsec-cr}, 
which means bounding the leakage
defined in \Cref{eq:leakage-cr} of \Cref{def:code-cr}.
We do this in general for the output of any channel, 
irrespective of actual encodings, and therefore will upper bound
a general leakage $\chi\qty(M;S,V\circ f_S^{-1})$.
We begin by converting the leakage 
from a Holevo quantity to a relative entropy.

\begin{lemma}
\label{lem:divasmi}
For any random variable $M$ over $\mathcal{M}$ 
independent of the uniform seed $S$, 
it holds
\begin{align*}
\chi(M;S,V\circ f_S^{-1}) \leq
\max_{m\in\mathcal{M}} \mathbb{E}_S D\qty( V\circ f_S^{-1}(m) \middle\| V(\mathcal{X})).
\end{align*}
\end{lemma}


\begin{proof}

Denote $\sigma_{s,m}:= V\circ f_s^{-1} (m)$.
We have
\begin{align*}
  &  \chi(M;S,V\circ f_S^{-1}) \\
&= \frac{1}{\mathcal{S}} \sum_s  S\left(\frac{1}{\mathcal{M}}\sum_m \sigma_{s,m}\right)
    - \frac{1}{\mathcal{S}}\frac{1}{\mathcal{M}}\sum_s \sum_m S\left( \sigma_{s,m}\right)
    \\
       &=S\left( \frac{1}{\mathcal{S}}  \frac{1}{\mathcal{M}}\sum_s \sum_m \sigma_{s,m}\right)
    - \frac{1}{\mathcal{S}}\frac{1}{\mathcal{M}}\sum_s \sum_m S\left( \sigma_{s,m}\right) - \chi(S;V\circ f_S^{-1})
    \\
    &= \chi(M,S;V\circ f_S^{-1})
    - \chi(S;V\circ f_S^{-1})\end{align*}
therefore,  
\begin{align*}
    \chi(M;S,V\circ f_S^{-1}) 
    &\leq \chi(M,S;V\circ f_S^{-1})\\
    &=  \mathbb{E}_M \mathbb{E}_S D\qty( V\circ f_S^{-1}(M) \middle\| 
        \mathbb{E}_M \mathbb{E}_S V\circ f_S^{-1}(M))\\
    &\leq \max_{m\in\mathcal{M}} \mathbb{E}_S D\qty( V\circ f_S^{-1}(m) \middle\| 
        \mathbb{E}_M \mathbb{E}_S V\circ f_S^{-1}(M)).
\end{align*} the equality holds because of \Cref{eq:holevo-avgrelent}, where we consider the 
classical-quantum channel $\mathcal{S}\times \mathcal{X}\rightarrow \mathcal{S}(H')$ that maps $s,m \to V\circ f_s^{-1} (m)$.
It only remains to show that 
\begin{equation*}
\mathbb{E}_M \mathbb{E}_S V\circ f_S^{-1}(M) = V(\mathcal{X}),
\end{equation*}
but this follows immediately from \Cref{eq:seed-uniformx},
namely from $\mathbb{E}_S V\circ f_S^{-1}(m) = V(\mathcal{X})$,
and the proof is concluded.
\end{proof}

For the next step we will define subnormalized classical-quantum channels. 
Later we will project onto the typical subspace,
and discard the rest.

\begin{definition}Let $\epsilon \geq 0$.
An $\epsilon$-subnormalized classical-quantum channel
$V': \mathcal{X}\to\mathcal{L}({\cal H})$
is a map
satisfying $V'(x)\geq 0$ and $1-\epsilon\leq \tr V'(x) \leq 1$, for all $x\in\mathcal{X}$.
Since for $\epsilon>\epsilon'$, an $\epsilon$-subnormalized channel is also $\epsilon'$-subnormal,
we call all $0$-subnormalized classical-quantum channels simply subnormal.
\\Now let $V:\mathcal{X}\to\mathcal{S}({\cal H})$ be a classical-quantum channel. Let  $V':\mathcal{X}\to\mathcal{L}({\cal H})$ be a subnormalized classical-quantum channel.
We say that 
\begin{equation*}
    V'\leq V
\end{equation*}
if $V'(x)\leq V(x)$ for all $x\in\mathcal{X}$.
\label{qnocsig4'2}
\end{definition}

The ordering definition reflects what we obtain when we project
a channel on a subspace, we obtain a subnormalized channel that is
less than the original channel in an operator ordering sense.
When we project onto the typical subspace, we only change the channel a little, 
and we want to make sure that our upper bound only changes a little.
This is the statement of the next lemma.

\begin{lemma}\label{lem:pass_to_subnormalization}Let $V:$ $\mathcal{X}$ $\to \mathcal{S}({\cal H})$ be a classical-quantum channel.
Let $\epsilon>0$. Let $V': \mathcal{X}\to\mathcal{L}({\cal H})$ 
be an $\epsilon$-subnormalized classical-quantum channel such that $V'\leq V$.
Then for any fixed $m\in\mathcal{M}$, it holds
\begin{equation*}
\mathbb{E}_S D\left(V \circ f_S^{-1} (m) \parallel V (\mathcal{X}) \right)\leq
\mathbb{E}_S D\left(V'\circ f_S^{-1} (m) \parallel V'(\mathcal{X}) \right)
+\epsilon\log\frac{|\mathcal{X}|}{d_{\mathcal{S}}}\text{ .}
\end{equation*}
\end{lemma}

\begin{proof}
Let $\rho$, $\rho'$, $\sigma$, and $\sigma'$ 
be subnormalized quantum states on the same system
such that the sum is a normalized state, 
namely so that $\tr(\rho+\rho')= \tr(\sigma+\sigma') =1$.
Consider classical-quantum states of the  form
$\ketbra{0} \otimes \rho + \ketbra{1} \otimes \rho'$ 
and $\ketbra{0} \otimes \sigma + \ketbra{1} \otimes \sigma'$.
By monotonicity of the relative entropy (cf. \Cref{BNaD}) under the trace, we note that
\begin{align*}
D&(\rho + \rho'\parallel  \sigma + \sigma')
\\&\leq 
D(\ketbra{0} \otimes \rho   + \ketbra{1} \otimes \rho' \parallel 
  \ketbra{0} \otimes \sigma + \ketbra{1} \otimes \sigma')
\\&= D(\rho \parallel \sigma )
   + D(\rho'\parallel \sigma').
\end{align*}

We define $V^\Delta\coloneqq V-V'$ and apply the above to
\begin{align*}
    \rho &\coloneqq V'\circ f_s^{-1} (m)
    &
    \rho'&\coloneqq V^\Delta\circ f_s^{-1} (m)
    \\
    \sigma &\coloneqq V' (\mathcal{X})
    &
    \sigma'&\coloneqq V^\Delta (\mathcal{X}),
\end{align*}
which satisfy $ V\circ f_s^{-1} (m)=\rho +\rho'$
and $V(\mathcal{X})= \sigma + \sigma'$, and obtain
\begin{align}
  &D\left(V\circ f_s^{-1} (m) \parallel V( \mathcal{X} ) \right) \nonumber
\\&\leq D\left(V'\circ f_s^{-1} (m) \parallel V'( \mathcal{X} ) \right)
+ D\left(V^{\Delta}\circ f_s^{-1} (m) \parallel V^{\Delta}( \mathcal{X} ) \right).
\label{eq:cqchannel-logsum}
\end{align}
Notice that 
\begin{align*}
    V^{\Delta}\circ f_s^{-1} (m)
    &=\frac{1}{d_{\mathcal{S}}}\sum_{{x\in f_s^{-1}(m)}} V^{\Delta}(x)
    \leq \frac{1}{d_{\mathcal{S}}} \sum_{x\in \mathcal{X}} V^{\Delta}(x)
    = \frac{|\mathcal{X}|}{d_{\mathcal{S}}} V^{\Delta}(\mathcal{X}).
\end{align*} 
This implies that $supp(V^{\Delta}\circ f_s^{-1}(m))$ $\subset supp(V^{\Delta}(\mathcal{X}))$
and by operator monotonicity of the logarithm~(cf.~\cite{Bh}), that $\log(V^{\Delta}\circ f_s^{-1}(m))$
$\leq \log(V^{\Delta}(\mathcal{X}))$.
Thus
\allowdisplaybreaks
\begin{align*}&
D\qty(V^{\Delta}\circ f_s^{-1} (m) \parallel V^{\Delta}(\mathcal{X}))\\
&=\tr[ V^{\Delta}\circ f_s^{-1} (m) \cdot 
\qty(\log V^{\Delta}\circ f_s^{-1} (m) 
    -\log V^{\Delta}(\mathcal{X}) )] \\
&\leq\tr[ V^{\Delta}\circ f_s^{-1} (m) \cdot 
\qty(\log \qty[ \frac{|\mathcal{X}|}{d_{\mathcal{S}}} V^{\Delta}( \mathcal{X})] 
    -\log V^{\Delta}(\mathcal{X}) )] \\
&=\tr[ V^{\Delta}\circ f_s^{-1} (m) \cdot \log\frac{|\mathcal{X}|}{d_{\mathcal{S}}} ]\\
&\leq \epsilon\log\frac{|\mathcal{X}|}{d_{\mathcal{S}}},
\end{align*}
where we used that $\tr V^{\Delta}\circ f_s^{-1} (m) \leq \epsilon$.
Plugging this into \Cref{eq:cqchannel-logsum}, we obtain the claim.
\end{proof}

\begin{lemma}
\label{lem:pass-to-renyi}
Let $\epsilon>0$ and let  $V': \mathcal{X}\to\mathcal{S}({\cal H})$ 
be an $\epsilon$-subnormalized classical-quantum channel.
For any fixed $m\in\mathcal{M}$, it holds
\begin{align*}
\mathbb{E}_S D\qty(V'\circ f_S^{-1} (m) \middle\| V'(\mathcal{X}))
&\leq
\log \mathbb{E}_S \exp\qty( D_2\qty(V'\circ f_S^{-1}(m)\middle\| V'(\mathcal{X})) ) +\epsilon.
\end{align*}
\label{d2vfmvux}
\end{lemma}
\Cref{d2vfmvux} is just the quantum version of Lemma 23 in~\cite{Wi/Bo}
and can be shown by similar techniques.
For the sake of completeness, we deliver a proof here.

\begin{proof} 
By the ordering relation and the convergence of $\alpha$-R\'enyi relative entropy 
for quantum states we mentioned in \Cref{BNaD}, 
it holds that $D(\rho\parallel\sigma) \leq  D_2(\rho\parallel\sigma)$,
and thus we can bound any relative entropy term $D(p\rho\parallel q\sigma)$,
where $p$ and $q$ are probabilities, and $\rho$ and $\sigma$ are states such that $supp(\rho) \subset supp(\sigma)$
with the $2$-R\'enyi  relative entropy as follows (notice that \Cref{eq:DtoD2} holds trivially if $supp(\rho)$ is not in $supp(\sigma)$):
{\allowdisplaybreaks
\begin{align*}
D(p\rho\parallel q\sigma)
  &= p D(\rho\parallel\sigma) + p \log \frac p q
\\&\leq p D_2(\rho\parallel\sigma) + p \log \frac p q
\\&= p \tr \log[\rho^2\sigma^{-1}] + p \log \frac p q
\\&= p \tr \log[(p\rho)^2(q\sigma)^{-1}] - p \log p
\\&= p D_2(p\rho\parallel q\sigma) + (1-p)
\\&\leq D_2(p\rho\parallel q\sigma) + (1-p),
\yesnumber
\label{eq:DtoD2}
\end{align*}
where we used $-p\log p \leq 1-p$.
}

We apply this to 
\begin{align*}
    p\rho &\coloneqq V'\circ f_s^{-1} (m)
    &
    q\sigma &\coloneqq V'(\mathcal{X})
\end{align*}
which gives 
\begin{align*}
    1-p &= 1 - \tr V'\circ f_s^{-1} (m) 
        \leq \epsilon 
\end{align*}
and thus
\begin{align*}
D(V'\circ f_s^{-1} (m) \parallel V'(\mathcal{X}))
&\leq D_2\qty(V'\circ f_s^{-1}(m)\parallel  V'(\mathcal{X}) ) + \epsilon\text{ .}
\end{align*} 
We apply the expectation value on both sides, and by the convexity of the exponential function we have
\begin{align*}
\mathbb{E}_S 
D(V'\circ f_s^{-1} (m) \parallel V'(\mathcal{X}))
&\leq\mathbb{E}_S D_2\qty(V'\circ f_S^{-1}(m)\middle\| V' (\mathcal{X}))  +\epsilon\\
&\leq\log \mathbb{E}_S\exp D_2\qty(V'\circ f_S^{-1}(m)\parallel V' (\mathcal{X}) ) +\epsilon\text{ .}
\end{align*} 
This concludes the proof.
%
%
\end{proof}

\begin{lemma}
\label{lvxshbcqc}
Let $V': \mathcal{X}\to\mathcal{L}({\cal H})$
be a subnormalized classical-quantum channel.
For every $m\in M$ we have
\begin{align*}
&\mathbb{E}_S\exp D_2\qty(V'\circ f_S^{-1} (m) \parallel V' (\mathcal{X})  )
\\&
\leq 
\lambda_2(f,m)  
 \rank[V'(\mathcal{X})] \max_{x\in \mathcal{X}} \norm{V'(x)}_\infty +1 
 ,
\yestag
\end{align*}
where $\lambda_2(f,m)$ is the second largest singular value 
of the $P_{f,m}$ defined in \Cref{bri:stochastic-matrix},
and the expectation value is over uniformly random $s\in\mathcal{S}$.
\end{lemma}

\Cref{lvxshbcqc} is a form of leftover hash lemma for classical-quantum channels.

\begin{proof}
Recall that if $supp(\rho)\subset supp(\sigma)$ then
$\exp D_2\qty(\rho\parallel \sigma ) = \tr\qty(\rho^2\sigma^{-1})$,
where $\sigma^{-1}$ is the pseudo-inverse.
By linearity, the definition of the BRI function, 
and by expanding the mixtures, 
for every $m$ we obtain
\begin{align*}
&   \mathbb{E}_S\exp D_2\qty(V'\circ f_S^{-1} (m) \parallel V' (\mathcal{X}) )\\
&=  \mathbb{E}_S \tr\qty[(V'\circ f_S^{-1} (m))^2 V' (\mathcal{X})^{-1} ] \\
&=  \frac{1}{|\mathcal{S}|}\sum_{s\in\mathcal{S}} \tr\qty[
    \frac{1}{d_{\mathcal{S}}^2} \sum_{x,x' \in \mathcal{X}}
    V'(x)\delta_{m,f_s(x)} \delta_{m,f_s(x')} V'(x') 
    (V' (\mathcal{X}))^{-1} ] \\
&=  \frac{d_{\mathcal{X}}}{|\mathcal{S}|d_{\mathcal{S}}} \sum_{x,x' \in \mathcal{X}}
    \tr\qty[V'(x)
    \frac{1}{d_{\mathcal{S}}d_{\mathcal{X}}} 
    \sum_{s\in\mathcal{S}} 
    \delta_{f_s(x),m} \delta_{f_s(x'),m} V'(x') 
    (V' (\mathcal{X}))^{-1} ] \\
&=  \frac{1}{|\mathcal{X}|} \sum_{x,x' \in \mathcal{X}}
    \tr\qty[V'(x) P_{f,m}(x,x') V'(x') 
    (V' (\mathcal{X}))^{-1} ] \\
\end{align*} 
{where we applied the expression of $P_{f,m}$ from \Cref{afsxnicabri},
and then \Cref{bri:sizes}.

Let now $\rho$ and $\sigma$ be two states with $supp(\rho)\subset supp(\sigma)$.
Let $\{\ket{v_i}\}$ be an orthonormal basis of eigenvectors to the non-zero
eigenvalues of $\sigma$ and
let $\rho_{ij} = \bra{v_i}\rho\ket{v_j}$. Then
\begin{align*}
\tr\qty(\rho^2\sigma^{-1})
= \sum_{ij} \rho_{ij} \rho_{ji} \sigma_{ii}^{-1}.
\end{align*}
Notice that $supp(V(x))\subset supp(V(\mathcal{X}))$ for all $x$, so we can apply this to the above.
Let  $\{\lambda_i:i\}$ be the non-zero
eigenvalues of  $V(\mathcal{X})$ with
a set of  orthonormal eigenvectors
$\{|v_i\rangle:i\}$.
We now use the notation ${V'}_{ij}$ for the functions
${V'}_{ij}(x) =  \bra{v_i}V'(x)\ket{v_j}$.
Because $V'(x)$ are Hermitian, we have ${V'}_{ij}={V'}^*_{ji}$, and thus we can write}
\begin{align*}
&\mathbb{E}_S\exp D_2\qty(V'\circ f_S^{-1} (m) \parallel V' (\mathcal{X}) )\\
&=  \frac{1}{|\mathcal{X}|} \sum_{x,x' \in \mathcal{X}} \sum_{ij}
    V'_{ij}(x) P_{f,m}(x,x') V'_{ji}(x') 
    \frac{1}{{V'}_{ii} (\mathcal{X})} \\
&=  \frac{1}{|\mathcal{X}|} \sum_{ij} \frac{1}{{V'}_{ii} (\mathcal{X})} 
    \bra{{V'}_{ij}} P_{f,m} \ket{V'_{ij}},
\yestag\label{eed2vfmp1}
\end{align*}
where $\ket{{V'}_{ij}}$ are complex vectors in $\mathbb{C}^\mathcal{X}$.

\newcommand\one{{\mathbf{1}}}
We now use the following well-known result (see, e.g.,~\cite{Bre}) 
(orignally stated for real vectors -- see \Cref{omPom} 
for the generalization of the proof to complex vectors). 
$P_{f,m}$ is a symmetric stochastic matrix in an $|\mathcal{X}|$ dimensional real space
with $\lambda_2(f,m)<1$ denoting the second-largest eigenvalue.
By construction and assumption, $1$ is the largest eigenvalue and it is simple (non-degenerate).
Then, for any two vectors in $\omega$ and $\omega'$ in this space, it holds that
\begin{equation}
\label{eq:second-eigenvalue}
\bra\omega P_{f,m} \ket\omega \leq \lambda_2(f,m) \bra\omega\ket{\omega}
+\bra\omega\ket{\one}\!\!\bra\one\ket{\omega},
\end{equation}
where $\ket\one$ is the the normalized all-one vector,
namely $\bra\one = \frac{1}{\sqrt{|\mathcal{X}|}} (1,\dots,1)$.

We thus have
\begin{align*}
&\mathbb{E}_S\exp D_2\qty(V'\circ f_S^{-1} (m) \parallel V' (\mathcal{X}) )\\
&\leq\sum_{ij} \frac
{|\mathcal{X}|^{-1}}
{{V'}_{ii}(\mathcal{X})} 
[\lambda_2(f,m) \bra{{V'}_{ij}}\ket{{V'}_{ij}}
+ \bra{{V'}_{ij}}\ket{\one}\bra{\one}\ket{{V'}_{ij}}]
\end{align*} 
However, notice that by contruction,
\begin{align*}
\bra{V'}_{ij}\ket{\one}\bra{\one}\ket{{V'}_{ij}}
= |\mathcal{X}| 
  \cdot {V'}_{ij}^*(\mathcal{X}) 
  \cdot {V'}_{ij}  (\mathcal{X})
= |\mathcal{X}| \delta_{ij} ({V'}_{ii}(\mathcal{X}))^2,
\end{align*}
because the choice of basis is an eigenbasis of $V'(\mathcal{X})$.
We can thus simplify the expression to
\begin{align*}
&\mathbb{E}_S\exp D_2\qty(V'\circ f_S^{-1} (m) \parallel V' (\mathcal{X}) )\\
&\leq
\frac{1}{|\mathcal{X}|}\lambda_2(f,m)
\sum_{ij} \frac
{\bra{{V'}_{ij}}\ket{{V'}_{ij}}}
{{V'}_{ii} (\mathcal{X})} 
+ \sum_i {V'}_{ii}(\mathcal{X})\\
&=
\frac{1}{|\mathcal{X}|}\lambda_2(f,m)
\sum_{ij} \frac
{\sum_{x} {V'}_{ij}(x) {V'}_{ji}(x)}
{{V'}_{ii} (\mathcal{X})} 
+ \tr [V'(\mathcal{X})]
\\
&=
\frac{1}{|\mathcal{X}|}\lambda_2(f,m)
\sum_{x} \tr[ {V'}^2(x) (V' (\mathcal{X}))^{-1} ]
+ 1.
\end{align*} 
Now we focus on $\sum_{x}\tr[ {V'}^2(x) (V' (\mathcal{X}))^{-1} ] $.
We repeatedly use the cyclic property of the trace and 
$\tr AB \leq \norm{A}_\infty \tr|B|$,
together with the positivity of the states, to obtain the following
{\allowdisplaybreaks
\begin{align*}
&\sum_{x}\tr[ {V'}^2(x) (V' (\mathcal{X}))^{-1} ] \\
&\leq \sum_{x} \norm{ V'(x)}_\infty 
\tr \qty(\sqrt{V'(x)}(V' (\mathcal{X}))^{-1} \sqrt{V(x)})\\
&\leq\max_{x\in \mathcal{X}} \norm{V'(x)}_\infty \sum_{x}  
\tr[ V'(x)(V' (\mathcal{X}))^{-1} ]\\
&=\max_{x\in \mathcal{X}} \norm{V'(x)}_\infty \cdot |\mathcal{X}| \cdot
\tr[ V' (\mathcal{X}) (V' (\mathcal{X}))^{-1} ]\\
&=\max_{x\in \mathcal{X}} \norm{V'(x)}_\infty \cdot |\mathcal{X}| \cdot
\rank (V' (\mathcal{X})).
\yestag
\label{eed2vfmp4}
\end{align*}}

Now we join everything together and obtain
\begin{align*}
&\mathbb{E}_S\exp D_2\bigl((V'\circ f_S^{-1}) (m) \parallel V'(\mathcal{X}) \bigr)
\\&
\leq \lambda_2(f,m)  
\rank[V' (\mathcal{X})] \max_{x\in \mathcal{X}} \norm{V'(x)}_\infty +1,
\yestag\label{eed2vfmp3}
\end{align*} 
as claimed.
\end{proof}

Note that the bound in \Cref{lvxshbcqc} is technically different from
the classical version in \cite[Lemma~26]{Wi/Bo}, which bounds the leakage in terms of a max mutual information. 
As for now, we are only able to prove the lemma for finite dimensional quantum systems, while the classical version is valid also for infinite classical systems.


\begin{remark}
Using different type of functions instead of BRI functions, Hayashi and Matsumoto show~\cite[Theorem~17 and Lemma~21]{Hay/Mat} which is a result similar to \Cref{lvxshbcqc}
in the case of a single message (i.e., for resolvability~\cite[Theorem~17]{Hay/Mat}) and for ordinary classical channels. 
It is straightforward to extend this to the case of several messages and subnormalized channels. The function class of Hayashi and Matsumoto is defined via the function inverses in terms of group homomorphisms. 
The example given in~\cite{Va/Ha} uses a short seed for strong security.
It is still open whether the seed required in \cite{Hay/Mat} for semantic secrecy  can be as short as for the BRI security functions in this work.  
The size of the seed is a part of complexity of the code once it is derandomized, it may partially influence the efficiency of the code and the finite rates achieved for a finite number of channel uses.
During the completion of this work, new efficient functions with an efficient randomness size were proven to achieve semantic security in the classical setting~\cite{WB21}.
We expect these functions to provide semantic security also for quantum channels.
\end{remark}

For completeness, the proof of \Cref{eq:second-eigenvalue} follows here.
Afterwards, we will continue the chain of inequalities and bound 
the $V$ dependent term of \Cref{lvxshbcqc}.
\begin{lemma}
\label{omPom}
\newcommand\one{{\mathbf{1}}}

Let $P$ be a symmetric stochastic real matrix
in a $|\mathcal{X}|$ dimensional complex space
and let the eigenvalue $1$ be simple. 
Denote the second-largest eigenvalue modulus $P$ by $\lambda_2$.
For every vector $\ket\omega$  in this space, it holds
\begin{equation*}
\bra\omega P \ket\omega \leq \lambda_2 \bra\omega\ket{\omega}
+\bra\omega\ket{\one}\!\!\bra\one\ket{\omega},
\end{equation*}
where $\ket\one$ is the the normalized all-one vector,
namely $\bra\one = \frac{1}{\sqrt{|\mathcal{X}|}} (1,\dots,1)$.
\end{lemma}

\begin{proof}
\allowdisplaybreaks
\newcommand\one{{\mathbf{1}}}
Notice that being doubly stochastic implies
satisfying $P\ket\one = \ket\one$ and $\bra\one P = \bra\one$.
We first add and remove the $\ket\one$ component from $\ket\omega$:
\begin{align*}
&\bra\omega P \ket\omega\\
&=(\bra\omega-\bra\omega\ket{\one}\!\!\bra\one)P(\ket\omega-\ket\one\!\!\bra\one\ket{\omega})
\\&\quad+\bra\omega\ket{\one}\!\!\bra\one P \ket\omega
+\bra\omega P \ket\one\!\!\bra\one\ket{\omega}
-\bra\omega\ket{\one}\!\!\bra{\one}P\ket{\one}\!\!\bra\one\ket{\omega}\\
&=(\bra\omega-\bra\omega\ket{\one}\!\!\bra\one)P(\ket\omega-\ket\one\!\!\bra\one\ket{\omega})
+\bra\omega\ket{\one}\!\!\bra\one\ket{\omega}\\
\intertext
{because $1$ is a simple eigenvalue,
in the remaining subspace the largest eigenvalue is $\lambda_2$.
Since $\lambda_2$ is positive for such a matrix~\cite{Bre}, we have }
&\bra\omega P \ket\omega\\
&\leq(\bra\omega-\bra\omega\ket{\one}\!\!\bra\one)\lambda_2
  (\ket\omega-\ket\one\!\!\bra\one\ket{\omega})
+\bra\omega\ket{\one}\!\!\bra\one\ket{\omega}\\
&=\lambda_2 (\bra\omega\ket{\omega}
-\bra\omega\ket{\one}\!\!\bra\one\ket{\omega})
+\bra\omega\ket{\one}\!\!\bra\one\ket{\omega}\\
&\leq \lambda_2 \bra\omega\ket{\omega}
+\bra\omega\ket{\one}\!\!\bra\one\ket{\omega}.
\end{align*}
\end{proof}

Putting all the results above together 
we obtain the following single statement for BRI functions.

\begin{corollary}
\label{thm:bri-function}
Let $\epsilon>0$ and let $V': \mathcal{X}\to\mathcal{L}({\cal H})$ 
be an $\epsilon$-subnormalized classical-quantum channel such that $V'\leq V$.
For any random variable $M$ over $\mathcal{M}$ 
independent of the uniform seed $S$, 
it holds
\begin{align*}
\chi(M;S,V\circ f_S^{-1}) 
&\leq
\frac{1}{\ln 2} \max_{m\in\mathcal{M}} \lambda_2(f,m)  
\rank[V'(\mathcal{X})] \max_{x\in \mathcal{X}} \norm{V'(x)}_\infty +
\\&\quad+\epsilon+\epsilon\log\frac{|\mathcal{X}|}{d_{\mathcal{S}}}.
\end{align*}
\end{corollary}
\begin{proof}
Joining \Cref{lem:divasmi,lem:pass_to_subnormalization,lem:pass-to-renyi,lvxshbcqc},
we obtain directly,
\begin{align*}
\chi(M;S,V\circ f_S^{-1}) 
&\leq
\log (1 + \max_{m\in\mathcal{M}} \lambda_2(f,m)  
\rank[V'(\mathcal{X})] \max_{x\in \mathcal{X}} \norm{V'(x)}_\infty )+
\\&\quad+\epsilon+\epsilon\log\frac{|\mathcal{X}|}{d_{\mathcal{S}}}.
\end{align*}
The result follows simply from $\log(1+x)\leq x/\ln 2$.
\end{proof}

In the next section we will finally define 
what a code using BRI functions looks like.
The chain of lemmas above will allow us to prove 
that we can achieve capacity with such codes.
There, the classical-quantum channels of the lemmas above
will be the classical-quantum channel generated by 
a transmission code around the actual wiretap channel.
So the $V$ above should not be confused with the actual wiretap,
but instead it will be the composition of the wiretap $V^{\otimes n}$ and the encoder.
This is why all the lemmas above are single letter.

\subsection{BRI modular codes}

We can now prove the final statements.
As will be noticed in the next proof, 
BRI functions are not used to upgrade the strong secrecy achieved 
by, e.g., a hash function or any strong-secrecy capacity-achieving code.
Instead, the BRI functions replace hash functions and 
directly produce a semantic-secrecy capacity-achieving code
out of a capacity-achieving error-correction code.

Let us now fix an actual wiretap channel.
We fix a finite space $\mathcal{X}$,
two finite quantum systems $H$ and $H'$,
and a classical-quantum wiretap channel $(W,V)$ 
defined as the classical-quantum channels $W:\mathcal{X}\to \mathcal{S}(H)$ and $V:\mathcal{X}\to \mathcal{S}(H')$.
For reference, recall that an $(n, |\mathcal{S}|, |\mathcal{M}|)$ common-randomness code is a finite subset 
$\qty{C^{s}=\{(E^{s}_m,D^{s}_m):m\in\mathcal{M}\}:s\in\mathcal{S}}$
of the set of $(n, |\mathcal{M}|)$ codes, labeled by a finite set $\mathcal{S}$, the common randomness.
We then define BRI modular codes as follows.

\begin{definition}
\renewcommand{\emph}{\textbf}
Let $\mathcal{S}$, and $\mathcal{M}$ be the finite sets 
for the space of the seeds, the messages and the encodings.
Let $\qty{x^n_c,D_c}_{c\in \mathcal{C}}$ 
be an $(n,|\mathcal{C}|)$ code for $W$,
and let $f:\mathcal{S}\times\mathcal{C}\to \mathcal{M}$ be a BRI function.
We define their \emph{BRI modular code} to be the common-randomness code such that
for every seed $s\in\mathcal{S}$ and message $m\in\mathcal{M}$
\begin{enumerate}
    \item the encoder $E^s_m$ is the uniform distribution over $\qty{x^n_c: c \in f_s^{-1}(m)}$,
    \item the decoder is $D^s_m = \sum_{c\in f_s^{-1}(m)} D_c$.
\end{enumerate}
\end{definition}
Notice that, in practice, for the decoder it will be more straightforward 
to simply decode $c$ and then compute directly $f_s(c)$,
instead of implementing the coarse grained decoding operators.

\begin{theorem}
\label{thm:bri-single}
\label{pdpaoimtcr} 
For any probability distribution $P$ over $\mathcal{X}$:
\begin{enumerate}
    \item there exist BRI modular codes achieving the semantic secrecy rate $\chi(P;W)-\chi(P;V)$ using codes achieving the transmission rate $\chi(P;W)$;
 \item the same rate is achievable with their derandomized codes.
\end{enumerate}
\end{theorem}

Notice that this theorem implies that the classical-quantum wiretap channel capacity can also be achieved with such modular codes, and thus in particular the second point of the theorem also provides an alternative proof for \Cref{thm:capacity} with BRI scheme in the case of finite dimensional channels. 
Indeed, since the theorem holds for all $P$, we can also directly achieve the supremum. 
This single letter formula then implies the multi-letter formula by standard argument.
More precisely, we can write the classical-quantum wiretap capacity as 
\begin{align}
    C_\textnormal{w} (W,V)&
    = \sup_{n\in\mathbb{N}}\frac{1}{n}\max_{ E}
    C_\textnormal{w}^1( EW^{\otimes n}, EV^{\otimes n}),
\end{align}
where 
\begin{equation}
    C_\textnormal{w}^1 (W,V) \coloneqq \max_{P}( \chi(P;W)- \chi(P;V)).
\end{equation}
Then the standard argument, which we reproduce in  \Cref{thm:bri} for completeness, shows that if $C_\textnormal{w}^1$ is achievable by a class of codes, it automatically follows that $C_\textnormal{w}$ is also achievable.

Finally, the finite blocklength results can be extracted by looking at \cref{eq:finite-blocklength:size,eq:bricode-leakage} in the proof, and they depend on the finite-blocklength parameters of the chosen transmission code.
Similar finite-blocklength results for derandomized codes are found in \cref{eq:finite-blocklength:derandomized}.

\begin{proof}
Fix the arbitrary distribution $P$, fix any $\epsilon>0$, and let $\delta$ be a positive number which will later be chosen as a function of $\epsilon$. By~\cite{Ho,Sch/Ni,Sch/Wes}, there exists a $\gamma>0$ such that for sufficiently large $n$, there exists an $(n,|\mathcal{X}'|)$ transmission code $\mathcal \{E'_{x'},D'_{x'}:{x'}\in\mathcal{X}'\}$ for $W$ whose rate is at least $\chi(P;W)-\epsilon/2$, whose maximal error probability is at most $2^{-n\gamma}$, and whose codewords  moreover are all $\delta$-typical, namely the encoders satisfy $E(\mathcal{T}_{P,\delta}^n\vert {x'})=1$ for all messages ${x'}\in\mathcal{X}'$. (For the definition of $\mathcal{T}_{P,\delta}^n$, see \Cref{WnbtVdto}. For an explicit proof that the error can be made to decrease exponentially, see, e.g.,~\cite[Lemma 4.1]{Hay}).

We have to find a suitable BRI function in order to ensure semantic security. By enlarging $n$ if necessary, we have enough flexibility to choose integers $k$ and $d$ satisfying
\begin{align*}
    n\left(\chi(P;V)+\frac{\epsilon}{4}\right)
    &\leq\log d
    \leq n\left(\chi(P;V)+\frac{\epsilon}{2}\right),\\
    n\left(\chi(P;W)-\frac{\epsilon}{2}\right)
    &\leq k+\log d
    \leq \log |\mathcal{X}'|.
\end{align*}
As previously mentioned, we can choose BRI functions $f:\mathcal S\times\mathcal S\to\mathcal M$ from~\cite{Wi/Bo}, satisfying $|\mathcal{S}|=2^kd$, $|\mathcal{M}|=2^k$ and 
\begin{equation}\label{eq:Ram_BRI_EV}
    \lambda_2(f,m) \leq \frac{4}{d}\leq 4\cdot  2^{-n\left(\chi(P;V)+\frac{\epsilon}{4}\right)} .
\end{equation}
We think of  $\mathcal S$ as a subset of $\mathcal{X}'$ and define $\qty{\mathcal{C}^s\coloneqq\qty{E^s_m,D_m^s:m\in\mathcal M}:s\in\mathcal S}$ to be the BRI modular code constructed from $\mathcal{C}' \coloneqq\qty{E'_s,D'_{s}:{s}\in\mathcal{S}}$ and $f$. Its rate clearly satisfies
\begin{equation}
    \frac{\log |\mathcal M|}{n} = \frac{k}{n} \geq\chi(P;W)-\chi(P;V)-\epsilon
    \label{eq:finite-blocklength:size}
\end{equation}
and the maximal error probability is no larger than that of the transmission code, i.e., it is upper-bounded by $2^{-n\gamma}$. In order to evaluate the security of the BRI modular code, we define the classical-quantum channel $U=E'V^{\otimes n}$ and upper-bound \[\chi(M;S,E^S V^{\otimes n}) = \chi(M;S,U\circ f_S^{-1}),\] for any random variable $M$ on $\mathcal M$ independent of the uniform seed $S$. 

We introduce a subnormalized classical-quantum channel $V':\mathcal T_{P,\delta}^n\to\mathcal L(\mathcal H^n)$ by defining
\begin{equation}
    \label{eq:bricode-tipical}
    {V'}(x^n)
    =\Pi^n_{PV, \delta}\Pi^n_{V,\delta}(x^n)\cdot V^{\otimes n}(x^n) \cdot \Pi^n_{V,\delta}(x^n) \Pi^n_{PV, \delta }.
\end{equation} 
For the definition of $\Pi^n_{PV, \delta}$ and $\Pi^n_{V,\delta}(x^n)$, see \Cref{WnbtVdto}. By \Cref{0dltsnpo}, $V'$ is a $2^{-n\eta(\delta)}$ subnormalized classical-quantum channel satisfying $V'(x^n)\leq V^{\otimes n}(x^n)$ for all  $x^n\in\mathcal T_{P,\delta}^n$. Since all codewords are contained in $\mathcal T_{P,\delta}^n$, 
\[U'\coloneqq E'V'\]
is a $2^{-n\eta(\delta)}$ subnormalized classical-quantum channel satisfying $U'\leq U$, and \Cref{thm:bri-function} and \Cref{eq:Ram_BRI_EV} imply
\begin{align}
\chi(M;S,E^S V^{\otimes n})
    &\leq\frac{4}{d\ln 2}\rank[U'(\mathcal{X}')]\max_{x'\in\mathcal{X}'}\norm{ U'(x')}_{\infty}+2^{-n\eta(\delta)}(k+1).
    \label{eq:ub_subn_div}
\end{align} 
Since the inputs are chosen from a set of  typical sequence $\mathcal{T}^n_{P,\delta}$,
\Cref{0dltsnpo} can be used to find
\begin{align*}
    \rank[U'(\mathcal{X}')]\max_{x'\in\mathcal{X}'}\lVert U'(x')\rVert_\infty
    &\leq\rank[V'(\mathcal X^n)]\max_{x^n\in T_{P,\delta}}\lVert V'(x^n)\rVert_\infty\\
    &\leq 2^{n(\chi(P,V)+\gamma''(\delta))}.
\end{align*}
Inserting this and \Cref{eq:Ram_BRI_EV} into \Cref{eq:ub_subn_div} gives
\begin{equation}
    \chi(M;S,E^S V^{\otimes n}) 
    \leq\frac{4}{\ln 2}2^{-n(\epsilon/4-\gamma''(\delta))}+(k+1)2^{-n\eta(\delta)}.
\end{equation}
Since $k$ is $n$ times the rate of our common-randomness code,
thus by \Cref{thm:capacity} it cannot grow faster than $nC_\textnormal{w}(W,V)$, and thus 
\begin{equation}
    \label{eq:bricode-leakage}
    \chi(M;S,E^S V^{\otimes n}) 
    \leq\frac{4}{\ln 2}2^{-n(\epsilon/4-\gamma''(\delta))}+(nC_\textnormal{w}(W,V)+\epsilon'+1)2^{-n\gamma}.
\end{equation}
Now choose $\delta$ small enough for $\gamma''(\delta)<\epsilon/4$ to hold. Then this upper bound tends to zero at exponential speed with $n$. Hence as the blocklength $n$ increases, our BRI modular code $\{\mathcal{C}^s:s\in\mathcal S\}$ achieves the rate $\chi(P;W)-\chi(P;V)$ with exponentially decreasing error probability and leakage.

\bigskip
As previously mentioned, the codes we constructed use common randomness.
This allows us to simply provide the seed needed by the BRI modular code and keep the proof focused on the properties of the BRI function. We now derandomize these codes. Notice, however, that this is a standard procedure, and does not really depend on the structure of the BRI modular codes, but simply in the scaling of its size and errors.

We derandomize, as per \Cref{def:derandomized}, the BRI modular code above. 
We set $n'=n$ and share the seed with the same transmission code $\mathcal{C}' $ used to construct the BRI modular code.
For the number of reuses of the seed we need to choose a sequence $(N(n))_{n\in\mathbb{N}}$ such
that $1 \ll N(n) \ll 2^{n\gamma}$,  $N(n) \ll 2^{n(\epsilon/4-\gamma''(\delta))}$, and
$N(n) \ll (nC_\textnormal{w}(W,V)+\epsilon'+1)^{-1}2^{n\gamma}$.
For simplicity it suffices to choose $N(n)=n-1$ and thus we define  $\bar{\mathcal{C}}$ as the $n-1$-derandomized code constructed from $\mathcal{C}'$ and $\qty{\mathcal{C}^s:s\in\mathcal S}$. The total number of channel uses is then $n^2$.
By \Cref{thm:Nderandom} we have 
an $(n^2, \frac{n^2-n}{n^2} R, 2^{-n\gamma}n)$ semantic-secrecy code.
Since now we have $|{\mathcal M}^{n-1}|= 2^{(n-1)k}$, similar to
 (\ref{eq:finite-blocklength:size}), we have
 \begin{equation}    \frac{\log |\mathcal M ^{n-1}|}{n^2} = \frac{n-1}{n}\frac{k}{n} \geq \qty(1-\frac{1}{n}) \qty(\chi(P;W)-\chi(P;V)-\epsilon).
\end{equation}
When we set  
 \Cref{eq:bricode-leakage} into (\ref{thm:Nderandomineq}),
we have
\begin{equation}
    \chi(M;\bar{E}V^{\otimes n^2}) 
    \leq\frac{4n}{\ln 2}2^{-n(\epsilon/4-\gamma''(\delta))}+n(nC_\textnormal{w}(W,V)+\epsilon'+1)2^{-n\gamma}.
\label{eq:finite-blocklength:derandomized}
\end{equation}
The same argument as for the BRI modular code now works. Since $\delta$ was chosen to satisfy $\gamma''(\delta)<\epsilon/4$, this upper bound still tends to zero with $n$, and our derandomized BRI modular code achieves the rate $\chi(P;W)-\chi(P;V)$.
\end{proof}

The following is the standard statement that any single letter achievable rate implies a multi-letter achievable rate. We give a proof for completeness.

\begin{lemma}
\label{thm:bri}
If $C^1 (W,V)$ is an achievable rate, then 
\begin{align}
    C^\infty (W,V) 
    \coloneqq \sup_{n\in\mathbb{N}}\frac{1}{n}\max_{ E}
    C^1( EW^{\otimes n}, EV^{\otimes n}),
\end{align}
where the maximum is over finite sets $\mathcal A$ and stochastic mappings $ E:\mathcal A\to\mathcal X^{n}$, is also an achievable rate.
\end{lemma}

\begin{proof}
In order to show that $C_{\text{sem}}(W,V)$ is also achievable given that $C^1( W, V)$ is achievable, we pick any $n$ and $ E$. We obtain a new classical-quantum wiretap channel $( EW^{\otimes n}, EV^{\otimes n})$ for which we know that the rate $C^1( EW^{\otimes n}, EV^{\otimes n})$ is achievable. Specifically, for any $\varepsilon>0$ and  sufficiently large $n'$, there exists an $(n',C^1( EW^{\otimes n}, EV^{\otimes n}) - \epsilon,\epsilon)$ code $\qty{E'_m,D'_m:m\in\mathcal{M}}$ for $( EW^{\otimes n}, EV^{\otimes n})$. 
The error and leakage only depend directly on the encoder and channels compositions $E({E}W^{\otimes n})^{\otimes n'} $ and $E({E}V^{\otimes n})^{\otimes n'} $, thus they do not change, and thus the code  $\qty{E'_m E^{n\smash'},D'_m:m\in\mathcal{M}}$ is a $(nn',(C^1( EW^{\otimes n}, EV^{\otimes n}) - \epsilon)/n,\epsilon)$ code for $(W,V)$. 
Therefore $C^1( EW^{\otimes n}, EV^{\otimes n}) /n$ is achievable for $(W,V)$.
Since the above holds for all $n$ and $E$, taking the supremum concludes the proof.
\end{proof}

In this section we showed that there exist modular coding schemes constructed from suitable transmission codes and BRI functions which achieve the security capacity of the classical-quantum wiretap channel and provide semantic security. Compared to the results of \Cref{sscocqwc}, the message sets of these modular codes are given explicitly via the BRI function. In particular, they do not depend on the wiretap channel.

\section{Further Perspectives}

In classical information, not only discrete channels, but also continuous channels are important subjects of study. In~\cite{Wi/Bo} semantic security was demonstrated for both discrete channels and continuous channels. Thus it will be very interesting to analyze if we can extend these results to continuous quantum channels.
As mentioned above, the results of~\cite{Wi/Bo} show how a non-secure code can be transformed into a semantic secure code.
Thus it will be a promising next step to analyze if these results can be extend to non-secure code for continuous quantum channels, e.g., classical-quantum Gaussian channels, which are continuous-variable classical-quantum channels undergoing a Gaussian-distributed thermal noise~\cite{Ho/So/Hi}. 
Furthermore, similar to the discrete channels, one can consider that
the eavesdropper will have access to the environment's final state~\cite{Ru/Ma} for
continuous quantum channels as well.
Thus it will be an interesting further step to analyze if the results of \Cref{totcwtraqc} can be extended to 
continuous quantum channels.
Further discussions will be the extension of these techniques on more complicated networks,
e.g., arbitrarily varying wiretap channels. This is currently also still open for classical networks.

\section*{Acknowledgments}

Holger Boche, Minglai Cai, Christian Deppe, and Roberto Ferrara were supported by the German Federal Ministry of Education and Research (BMBF) through Grants 16KISQ028 (Deppe, Ferrara), 16KISQ020 (Boche), 16KIS0948 (Boche, Wiese) and 16KISQ038 (Deppe, Ferrara).
We thank the research hub 6G-life under Grant 16KISK002 for their support of Holger Boche and Christian Deppe. Holger Boche and Moritz Wiese were supported by the German Research Foundation (DFG) within the Germany’s Excellence Strategy - EXC 2092 CASA- 390781972.

\appendix
\setcounter{equation}{0}
\renewcommand{\theequation}{\thesection\arabic{equation}}

\section{Technical Lemmas}%
\label{WnbtVdto}

We now bound the $V$ dependent term of \Cref{lvxshbcqc}. Before stating the actual lemma, we need to recall some facts about typical sequences and typical operators, as can be found, e.g., in~\cite{Wil}.

Let $\mathcal{X}$ be a finite set.
Let $P$ be a probability function on
$\mathcal{X}$.
Let $\delta > 0$ and $n \in \mathbb{N}$. The set $\mathcal T_{P,\delta}^n$ of typical sequences of $P$ consists of those $x^n\in\mathcal X^n$ satisfying
\begin{itemize}
    \item $\displaystyle \qty| \frac{1}{n} N(a\mid x^n) - P(x') | \leq \frac{\delta}{|\mathcal{X}|}$ for all $a\in \mathcal{X}$,
    \item $N(a\vert x^n)=0$ if $P(a)=0$ for all $a$ in $\mathcal X$,
\end{itemize}
where $N(a\mid x^n)$ is the number of occurrences of the symbol $a\in\mathcal X$ in the sequence $x^n$. 

Let $H$ be a finite-dimensional
complex Hilbert space.
Let $\rho \in \mathcal{S}(H)$ be a state 
with spectral decomposition $\rho = \sum_{x} P(x)    \ketbra{x}$.
For any other basis of eigenvectors the same statements will be valid.
The $\delta$-typical subspace is defined as the subspace spanned
by $\qty{\ket{x^n}: x^n \in {\mathcal{T}}^n_{P, \delta}}$,
where  $\ket{x^n}\coloneqq\bigotimes_{i=1}^n \ket{x_i}$.  
The orthogonal projector onto the $\delta$-typical subspace is
\[\Pi^n_{\rho,\delta}=\sum_{x^n\in {\mathcal{T}}^n_{P,\delta}} \ketbra{x^n}\text{ .}\]
and satisfies the following properties.
There are positive constants $\alpha(\delta)$, $\beta(\delta)$, 
and $\gamma(\delta)$, depending on $\delta$ such that
for large enough $n$
\begin{align} 
\label{te1}
&\quad \tr({\rho}^{\otimes n} \Pi^n_{\rho ,\delta}) 
> 1-2^{-n\alpha(\delta)}
\text{ ,}
\\
\label{te2}
2^{n(S(\rho)-\beta(\delta))}
&\le \tr (\Pi^n_{\rho ,\delta})
\le 2^{n(S(\rho)+\beta(\delta))}
\text{ ,}
\\
\label{te3}
2^{-n(S(\rho)+\gamma(\delta))} \Pi^n_{\rho ,\delta} 
&\le \Pi^n_{\rho ,\delta} \cdot {\rho}^{\otimes n} \cdot \Pi^n_{\rho ,\delta}
\le 2^{-n(S(\rho)-\gamma(\delta))} \Pi^n_{\rho ,\delta}
\text{ .}
\end{align}

Similarly let $\mathcal{X}$ be a finite set, and  $H$ be a finite-dimensional
complex Hilbert space.
Let
$V:\mathcal{X}\to
\mathcal{S}(H)$ be a classical-quantum
channel. For $a\in \mathcal{X}$  suppose
$V(a)$ has
the spectral decomposition $V(a)$
$ =$ $\sum_{j}
V(j|a) |j\rangle\langle j|$
for a stochastic matrix
$V(\cdot|\cdot)$.
 The $\alpha$-conditional typical
subspace of $V$ for a typical sequence   $a^n$ is the
subspace spanned by
 $\left\{\bigotimes_{a\in\mathcal{X}} \ket{j^{\mathtt{I}_a}}, j^{\mathtt{I}_a} \in \mathcal{T}^{\mathtt{I}_a}_{V(\cdot|a),\delta}\right\}$.
Here $\mathtt{I}_a$ $:=$ $\{i\in\{1,\cdots,n\}: a_i = a\}$ is an indicator set that selects the indices $i$ in the sequence $a^n$
$=$ $(a_1,\cdots,a_n)$ for which the $i$-th
symbol $a_i$ is equal to $a\in\mathcal{X}$.
The subspace is often referred to as the $\alpha$-conditional typical
subspace of the state  $V^{\otimes n}(a^n)$.
The orthogonal subspace projector which projects onto it is defined as
\[\Pi^n_{V, \alpha}(a^n)=\bigotimes_{a\in\mathcal{X}}
\sum_{j^{\mathtt{I}_a} \in {\cal
T}^{\mathtt{I}_a}_{V(\cdot \mid a^n),\alpha}}|j^{\mathtt{I}_a} \rangle\langle j^{\mathtt{I}_a}|\text{ .}
\]

For 
$a^n \in {\mathcal{T}}^n_{P, \alpha}$ 
there are positive constants $\beta(\alpha)'$, $\gamma(\alpha)'$, 
and $\delta(\alpha)'$, depending on $\alpha$ such that
\begin{gather} 
\label{te4}
\tr\left({V}^{\otimes n}(x^n) \Pi^n_{{V}, \delta}(x^n)\right)
> 1-2^{-n\alpha'(\delta)}
\text{ ,}
\\
\label{te6}
2^{n(S({V}|P)-\beta'(\delta))}\le \tr\left(
\Pi^n_{{V}, \delta}(x^n) \right)\le 2^{n(S({V}|P)+\beta'(\delta))}\text{ ,}
\\
\label{te5}
\begin{aligned}
2^{-n(S({V}|P)+\gamma'(\delta))} \Pi^n_{{V}, \delta}(x^n)
 &\le \Pi^n_{{V}, \delta}(x^n){V}^{\otimes n}(x^n) \Pi^n_{{V}, \delta}(x^n)\\
 &\le 2^{-n(S({V}|P)-\gamma'(\delta))} \Pi^n_{{V}, \delta}(x^n)\text{ .}
\end{aligned}
\end{gather}

For the classical-quantum channel
$V:\mathcal{X}\to
\mathcal{S}(H)$  and a probability
distribution $P$ on $\mathcal{X}$  we define
 a quantum state $PV$ $:=$ $\sum_{a}P(a) V(a)$ on $\mathcal{S}(H)$.
Clearly, one can then speak of the
orthogonal subspace projector $\Pi_{PV, \delta}$
 fulfilling \Cref{te1,te2,te3}.
For $\Pi_{PV, \delta}$ there is a positive constant $\alpha(\delta)''$ such that for every $x^n\in{\mathcal{T}}^n_{P, \delta}$, the
following inequality holds:
\begin{equation} 
\label{te7}  
\tr \left(  {V}^{\otimes n}(x^n) \cdot \Pi^n_{P{V}, \delta } \right)
\geq 1-2^{-n\alpha''(\delta)} \text{ .}
\end{equation}

\begin{lemma}
\label{0dltsnpo}
Let $V: \mathcal{X}\to\mathcal{S}({\cal H})$ be a  classical-quantum channel. For any $\delta>0$ and probability distribution $P$ over $\mathcal{X}$, define the subnormalized classical-quantum channel $V':\mathcal T_{P,\delta}^n\to\mathcal L(\mathcal H^n)$ by
\begin{equation}
{V'}(x^n)\coloneqq
\Pi^n_{PV, \delta }    
\Pi^n_{V,\delta}(x^n)
\cdot V^{\otimes n}(x^n) \cdot 
\Pi^n_{V,\delta}(x^n) 
\Pi^n_{PV, \delta }.
\end{equation}
We assume that the inputs are chosen from a set of  typical sequence $\mathcal{T}^n_{P,\delta}$ with a probability distribution $P$ and a positive $\delta$.
Then $V'\leq V^{\otimes n}$. Moreover, there exist positive $\eta(\delta)$ and $\gamma''(\delta)$ such that if $n$ is sufficiently large, $V'$ is a $2^{-n\eta(\delta)}$-subnormalized classical-quantum channel and
\begin{equation} 
\rank[V' (\mathcal{X}^n)] \max_{x^n\in \mathcal{T}^n_{P,\delta}} \|V'(x^n)\|_\infty
 \leq 2^{n\chi(P;V)+n\gamma''(\delta) }.
\end{equation}
\end{lemma}

\begin{proof}
It is obvious that $V'\leq V^{\otimes n}$. To check that the trace of $V'$ is close to 1, let $x^n\in\mathcal T_{P,\delta}^n$ and define 
\[
    V''(x^n)=\Pi_{V,\delta}^n(x^n)V^{\otimes n}\Pi_{V,\delta}^n(x^n).
\]
Clearly
\[
    \tr(V'(x^n))
    =\tr(V''(x^n))-\tr((I-\Pi_{PV,\delta}^n)V''(x^n)).
\]
By \Cref{te4}, $\tr(V''(x^n))\geq 1-2^{-n\alpha'(\delta)}$. Also, it is clear that $V^{\otimes n}(x^n)$ commutes with $\Pi_{V,\delta}^n(x^n)$ and that $V''(x^n)\leq V^{\otimes n}(x^n)$. Therefore 
\begin{align*}
    \tr((I-\Pi_{PV,\delta}^n)V''(x^n))
    &\leq\tr((I-\Pi_{PV,\delta}^n)V^{\otimes n}(x^n))\\
    &\leq 2^{-n\alpha''(\delta)}.
\end{align*}
Altogether, if we set
\[
    \eta(\delta)=\min\{\alpha'(\delta),\alpha''(\delta)\},
\]
we obtain that
\[
    \tr(V'(x^n))\geq1-2^{-n\eta(\delta)},
\]
so $V'$ is a $2^{-n\eta(\delta)}$-subnormalized version of $V^{\otimes n}$.


Now we bound 
$ \max_{x^n\in \mathcal{T}^n_{P,\delta}}\|V'(x^n)\|_\infty$.
By \Cref{te5}, for any $x^n\in \mathcal{T}^n_{P,\delta}$ we have
\begin{align}
{V'}(x^n)
&\leq
\Pi^n_{V,\delta}(x^n)
\cdot V^{\otimes n}(x^n) \cdot 
\Pi^n_{V,\delta}(x^n) 
\\&\leq
2^{-n(S(V|P)-\gamma'(\delta))}
\Pi^n_{V,\delta}(x^n) 
\text{ ,}
\end{align} 
thus implying that 
\begin{align}
\max_{x^n\in \mathcal{T}^n_{P,\delta}} \norm{{V'}(x^n)}_\infty
&\leq
2^{-n(S(V|P)-\gamma'(\delta))}.
\label{vxnpvp3} 
\end{align}

We finally bound the $\rank [V'(\mathcal{X}^n) ]$.
By \Cref{te2} we have
\begin{align}
\rank [V'(\mathcal{X}^n) ]
&\leq \rank [\Pi^n_{PV, \delta} ]
= \tr\Pi^n_{PV, \delta} 
\leq 2^{n (S(PV)+\beta(\delta))}
\text{ .}
\label{vxnpvp4} 
\end{align}

We combine \Cref{vxnpvp3,vxnpvp4} and obtain
\begin{equation}
\rank[V'(\mathcal{X})] \max_{x^n\in T_{P,\delta}^n} \norm{V'(x^n)}_\infty  
\leq 2^{n(\chi(P;V)+\beta(\delta) +\gamma'(\delta))}   \text{ .}
\label{vxnpvp5}
\end{equation} 
\end{proof}

\end{document}